\documentclass[preprint]{elsarticle}
\usepackage{amsmath, amssymb,amsthm}
\usepackage{color}
\usepackage{latexsym}
\usepackage{indentfirst}
\usepackage{graphicx}
\usepackage{placeins}
\usepackage{booktabs}
\usepackage{algorithm}
\usepackage{algorithmic}
\usepackage{multirow}
\usepackage{subfig}

\biboptions{compress}

\newtheorem{theorem}{Theorem}

\newtheorem{prop}[theorem]{Proposition}

\DeclareMathOperator{\Tr}{Tr}

\newcommand{\ie}{\emph{i.e.}{}}
\newcommand{\eg}{\emph{e.g.}{}}

\newcommand{\CC}{\mathbb{C}}
\newcommand{\Or}{\mathcal{O}}

\newcommand{\norm}[1]{\left\lVert#1\right\rVert}

\newcommand{\I}{\imath} 

\newcommand{\barint}{\kern4pt \raise3.4pt\hbox{\vrule height.6pt
    width7pt} \kern-11pt \int}

\newcommand{\bvec}[1]{\mathbf{#1}}

\newcommand{\vc}{\bvec{c}}
\newcommand{\ve}{\bvec{e}}
\newcommand{\vj}{\bvec{j}}
\newcommand{\vk}{\bvec{k}}
\newcommand{\vn}{\bvec{n}}

\newcommand{\vr}{\bvec{r}}
\newcommand{\vs}{\bvec{s}}

\newcommand{\CS}{\mathcal{C}}
\newcommand{\FS}{F}
\newcommand{\KS}{\mathcal{K}} 
\newcommand{\MS}{M}
\newcommand{\NS}{N}
\newcommand{\RS}{\mathcal{R}} 

\newcommand{\Rone}[1]{{#1}}
\newcommand{\Rtwo}[1]{{#1}}

\journal{Journal of Computational Physics}

\begin{document}

\begin{frontmatter}

\title{SCDM-k: Localized orbitals for solids via selected columns of the density matrix}

\author[ucb]{Anil Damle\corref{cor}}
\ead{damle@berkeley.edu}
\cortext[cor]{Corresponding author}

\author[ucb,lbl]{Lin Lin}
\ead{linlin@math.berkeley.edu}

\author[smath,icme]{Lexing Ying}
\ead{lexing@math.stanford.edu}

\address[ucb]{Department of Mathematics, University of California,
Berkeley, Berkeley, CA 94720}

\address[lbl]{Computational Research Division, Lawrence Berkeley National
Laboratory, Berkeley, CA 94720}

\address[icme]{Institute for Computational and Mathematical Engineering,
Stanford University, Stanford, CA 94305}

\address[smath]{Department of Mathematics, Stanford University, Stanford, CA 94305}

\begin{abstract}
  The recently developed selected columns of the density matrix (SCDM)
  method [J. Chem. Theory Comput. 11, 1463, 2015] is a simple, robust,
  efficient and highly parallelizable method for constructing localized
  orbitals from a set of delocalized Kohn-Sham orbitals for insulators
  and semiconductors with $\Gamma$ point sampling of the Brillouin zone.
  In this work we generalize the SCDM method to Kohn-Sham density
  functional theory calculations with $\vk$-point sampling of the
  Brillouin zone, which is needed for more general electronic structure
  calculations for solids. 
  We demonstrate that our new method, called SCDM-k, is by construction
  gauge independent and a natural way to describe localized orbitals.
  SCDM-k computes localized orbitals without the use of an optimization
  procedure, and thus does not suffer from the possibility of being
  trapped in a local minimum.  Furthermore, the computational complexity
  of using SCDM-k to construct orthogonal and localized orbitals scales
  as $\Or(\NS\log \NS)$ where $\NS$ is the total number of $\vk$-points
  in the Brillouin zone. SCDM-k is therefore efficient even when a large
  number of $\vk$-points are used for Brillouin zone sampling.  We
  demonstrate the numerical performance of SCDM-k using systems with
  model potentials in two and three dimensions.
\end{abstract}


\end{frontmatter}


\section{Introduction}

Kohn-Sham density functional theory
(DFT)~\cite{HohenbergKohn1964,KohnSham1965} is the most widely used
electronic structure theory for molecules and systems in condensed
phase. The Kohn-Sham orbitals (a.k.a. Kohn-Sham wavefunctions) are eigenfunctions of the Kohn-Sham
Hamiltonian. We refer to the span of a given set
of Kohn-Sham orbitals as the Kohn-Sham invariant subspace.  These orbitals are in general delocalized, \ie~each orbital
has significant magnitude across the entire
computational domain. However, information about atomic structure and chemical
bonding, which is often localized in real space, may be difficult
to interpret from delocalized Kohn-Sham orbitals. The connection
between localized and delocalized information is made possible by
a \textit{localization} procedure. 

A
localization procedure finds a set of orbitals that
are localized in real space, and span the Kohn-Sham invariant subspace. Examples of widely
used localization schemes include Boys
localization~\cite{FosterBoys1960} mostly in the context of chemistry,
and maximally localized Wannier functions
(MLWFs)~\cite{MarzariVanderbilt1997,WannierReview} mostly in the context of
physics and materials science. The localized orbitals are not only
useful for analyzing chemical and materials systems, but can also serve
as powerful computational tools for hybrid functional
calculations~\cite{WuSelloniCar2009,GygiDuchemin2012}, theory of
polarization of crystalline solids based on Berry-phase
calculations~\cite{King-SmithVanderbilt1993}, interpolation of band
structure~\cite{MarzariVanderbilt1997}, linear scaling DFT
calculations~\cite{Goedecker1999}, and excited state
theories~\cite{UmariStenuitBaroni2009,UmariStenuitBaroni2010} among
others. Because of the wide range of applications for localized
orbitals, several other localization methods have also been proposed in
the past few
years~\cite{Gygi2009,ELiLu2010,OzolinsLaiCaflischEtAl2013,AquilantePedersenMerasEtAl2006,SCDM}.

The potential for constructing localized orbitals from delocalized
Kohn-Sham orbitals can be justified physically by the
``nearsightedness'' principle for electronic matter of finite HOMO-LUMO
gap~\cite{Kohn1996,ProdanKohn2005}. The nearsightedness principle can be
more rigorously stated as the single particle density matrix (DM)
being exponentially localized along the off-diagonal direction in its real
space
representation~\cite{BenziBoitoRazouk2013,Kohn1996,Blount,Cloizeaux1964a,Cloizeaux1964b,
Nenciu}. Based on the exponential decay of the DM in the real space, we
have recently developed the selected columns of the density matrix
(SCDM) method~\cite{SCDM} as a new way to
construct localized orbitals. The method is simple, robust, efficient
and highly parallelizable. As the name suggests, the localized orbitals
are obtained directly from a column selection procedure implicitly
applied to the density matrix. Hence, the locality of these columns is a
direct consequence of the locality of the density matrix. In contrast
with Boys localized orbitals or MLWFs, our method does not attempt to
minimize a given localization measure via a minimization procedure.
Consequently, our method does not require any initial guess of localized
orbitals, and its cost is predetermined for a given problem size. It
also avoids some of the potential problems associated with a
minimization scheme, such as getting stuck at a local minimum.  

\Rtwo{For isolated molecules, the number of electrons is relatively
small. On the other hand, the number of electrons in solids can reach
macroscopic scale, and the calculation must be simplified. Using the
fact that the potential and the electron density are periodic with
respect to the unit cell of a solid system, one can perform a Bloch
decomposition of the Kohn-Sham Hamiltonian. The wavefunctions for each
Bloch decomposed Hamiltonian satisfy twisted boundary conditions indexed
by a vector $\vk$ belonging to the so-called Brillouin zone. In order to compute
physical quantities such as the electron density and total energy in
Kohn-Sham DFT, the Brillouin zone needs to be represented by a number of
discrete $\vk$-points. This is called Brillouin zone sampling.  We
refer readers to section 3 as well as ~\cite{Martin2004} for more details of the Bloch
decomposition and Brillouin zone sampling. In particular, the scheme
using one special $\vk$ point, denoted by $\Gamma=(0,0,0)^T$, to sample the
Brillouin zone is called the $\Gamma$ point sampling scheme.}
The SCDM procedure proposed in Ref.~\cite{SCDM} is
applicable to Kohn-Sham DFT calculations for isolated molecules,
and for solids with $\Gamma$ point sampling of the Brillouin zone. 

In many physics and materials science applications such as chemical bonding
analysis of complex solids, band structure interpolation, and Berry-phase
theories, localized orbitals need to be constructed from Kohn-Sham
orbitals obtained from a set of $\vk$-points in the Brillouin zone other
than the $\Gamma$ point. \Rtwo{The number of $\vk$ points needed 
is system dependent, and can range from tens to tens of thousands.}
The
common practice for Brillouin zone sampling is to diagonalize the
Kohn-Sham Hamiltonian matrix for each $\vk$-point independently. Since
the Kohn-Sham Hamiltonian matrix is in general complex Hermitian, the
Kohn-Sham orbitals obtained for each $\vk$-point can acquire an
arbitrary phase, often referred to as the ``gauge'' of the orbitals. For
degenerate orbitals (i.e. orbitals with the same eigenvalue) the gauge
can be an arbitrary unitary matrix.
The widely used method for finding MLWFs~\cite{WannierReview} is
gauge-dependent. It involves the differentiation operator with respect to
the Brillouin zone index $\vk$. Therefore a gauge transformation needs
to be performed prior to the minimization procedure to smooth the gauge, so
that the differentiation operator is well
defined~\cite{MarzariVanderbilt1997}. Such a gauge smoothing procedure is
not unique. After the gauge transformation, the computation of MLWFs
requires the minimization of a nonlinear, non-convex energy functional. Therefore,
the iterative procedure may get stuck at local minimum. Furthermore, the
nonlinear energy functional and its solution can depend heavily on the
initial guess.  This is especially the case for materials where within
the unit cell there is complex atomic structure.

In this paper, we generalize the SCDM procedure for finding localized
orbitals to solids with $\vk$-point sampling.  The new method, which
we refer to as SCDM-k, has a few notable features. First, the
localized orbitals are obtained directly from columns of the density
matrix, which is a gauge invariant quantity. Thus, SCDM-k does not
require a gauge transformation, and the result is independent of the
choice of the gauge. Second, SCDM-k is a direct method that does not
involve an iterative optimization procedure and thus avoids getting
stuck at a local minimum. Third, SCDM-k has only one parameter to
adjust (size of the local supercell), which is introduced to
improve the efficiency, and our numerical experiments indicate that the
quality of the localized orbitals is relatively insensitive to the
choice of this parameter. Finally, the SCDM-k procedure is highly
efficient. The complexity for generating non-orthogonal and orthogonal
localized orbitals is $\Or(\NS)$ and $\Or(\NS\log \NS)$, respectively
where $\NS$ is the total number of $\vk$-points in the Brillouin zone.
Therefore the method is suitable even when a large number of
$\vk$-points are used for sampling the Brillouin zone.

The paper is organized as follows. Section \ref{sec:prelim} outlines
the notation and some concepts that will be used throughout this paper. Section
\ref{sec:KS} then outlines the procedure for solving Kohn-Sham DFT with
Brillouin zone sampling. After briefly reviewing the SCDM procedure for
the $\Gamma$ point case, we describe in section~\ref{sec:scdm} 
the SCDM-k procedure for  $\vk$-point sampling. Finally, section
\ref{sec:numer} presents numerical results in two and three dimensions
using a model potential and is followed by concluding remarks and future
directions in section \ref{sec:conclusion}.

\section{Preliminaries}
\label{sec:prelim}
\subsection{Notation} \label{sec:notation}
A relatively self-contained discussion of $\vk$-point sampling requires 
the introduction of a considerable amount of notation. 
Table \ref{tab:nomen} summarizes the requisite notation that we will be
using throughout this manuscript. We also provide a brief overview of
some of the more pervasive notation used throughout the rest of the
text, and introduce the remainder as needed.  In the discussion below,
without loss of generality we assume the dimension $d=3$, and the
formulation can be easily extended to $d=1$ or $d=2$.

We denote a unit cell by $\Omega^{u}$.  The global supercell, denoted by
$\Omega^{g}$, contains $N_{1}\times N_{2}\times N_{3}$ unit cells
equipped with periodic boundary conditions. 
Due to the translational
symmetry, the problem on a global supercell can be
equivalently decomposed into $N_{1}\times N_{2}\times N_{3}$
\textit{independent} problems on a unit cell, each represented by a
$\vk$-point in the Brillouin zone using a Monkhorst-Pack
grid~\cite{MonkhorstPack1976}.  One important component of the SCDM-k
method is the so-called local supercell $\Omega^\ell$ associated with
the unit cell $\Omega^{u}$. A local supercell
is comprised of $N^\ell_{1}\times N^\ell_{2}\times N^\ell_{3}$ adjacent
unit cells  ($N^{\ell}_{i}\le N_{i},i=1,2,3$). Figure \ref{fig:kpoint}
illustrates the described relationship between the unit cell, local
supercell, and global supercell in a two dimensional setting.  

Computationally, each $\vk$-point problem can be solved with any
suitable discrete basis set. Below we assume a planewave basis set is
used with $M_{1}\times M_{2}\times M_{3}$ grid points in reciprocal
space. This corresponds to a set of grid points of the same size in real
space discretizing $\Omega^{u}$ uniformly. With a slight abuse of
notation, this set of discrete grid points in the unit cell is also denoted by
$\Omega^{u}$. A similar abuse of notation is used for the global supercell
$\Omega^{g}$ and the local supercell $\Omega^{\ell}$. 

\Rtwo{To present the algorithms generally, we allow for distinct numbers
of points in each of the three dimensions. This is the case for both the
real space grid of the unite cell, and the $\vk$-point grid. However,
often the asymptotic computational cost will only be a function of the
product of the number of points per dimension. Therefore, we use capital
letters with subscripts such as $N_{1},N_{2},N_{3}$ to define the number
of points per dimension and the same capital letter sans subscript such
as $N$ to denote the total number of points. In addition, we use
calligraphic letters such as $\KS$ to denote sets. These will be used for operations such as general indexing of matrices or summation.}

\begin{table}
  \centering
  \begin{tabular}{l|p{8cm}}
     \hline
     \hline
     Notation & Description\\
     \hline
     $\Omega^{u}$ & Unit cell; Collection of indices corresponding to
     uniform real space grid points in the unit cell\\
     \hline
     $\Omega^{\ell}$ & Local supercell; Collection of indices
     corresponding to uniform real space grid points in a local supercell\\
     \hline
     $\Omega^{g}$ & Global supercell; Collection of indices
     corresponding to uniform real space grid points in a global supercell\\
     \hline
     \hline
     $\I$ & Imaginary unit $\sqrt{-1}$ \\
     \hline
     $\ve_{1},\ve_{2},\ve_{3}$ & Unit vector along each dimension\\
     \hline
     $\vs=(s_1,s_2,s_3)$ & Shift vector of the Monkhorst-Pack grid\\
     \hline
     $\vk=(k_1,k_2,k_3)$ & A $\vk$-point \\
     \hline
     $M_{1},M_{2},M_{3}$ & Number of uniform grid points in the unit
     cell along each dimension\\
     \hline
     $N_{1},N_{2},N_{3}$ & Number of $\vk$-points for Brillouin zone
     sampling along each dimension\\
     \hline
     $N_{1}^{\ell},N_{2}^{\ell},N_{3}^{\ell}$ & Number of unit cells
     in a local supercell along each dimension\\
     \hline
     $L_{1},L_{2},L_{3}$ & Length of the unit cell $\Omega^{u}$ along
     each dimension \\
     \hline
     $\MS$ & $M_{1}\times M_{2}\times M_{3}$\\
     \hline
     $\NS$ & $N_{1}\times N_{2}\times N_{3}$\\
     \hline
     $\NS^{\ell}$ & $N^{\ell}_{1}\times N^{\ell}_{2}\times N^{\ell}_{3}$\\
     \hline
     $\KS$ & Collection of all the $\vk$-points of the Monkhorst-Pack grid corresponding to the
     global supercell\\
     \hline
     $\KS^\ell$ & Collection of all the $\vk$-points of the Monkhorst-Pack grid corresponding to a
     local supercell\\
     \hline
     $n_{b}$ & Number of wavefunctions in the unit cell\\
     \hline
     $\psi_{b,\vk}$ &  A Kohn-Sham orbital on the global supercell
     $\Omega^{g}$, both at the continuous and the discrete level\\
     \hline
     $\psi^{\ell}_{b,\vk}$ &  A Kohn-Sham orbital on the
     local supercell $\Omega^{\ell}$, both at the continuous and the discrete level\\
     \hline
     $u_{b,\vk}$ &  Periodic part of a Kohn-Sham orbital on the unit
     cell $\Omega^{u}$, both at the continuous and the discrete level\\
     \hline
     $P_{\vk}$ & Density matrix corresponding a particular $\vk$ point
     on the global supercell $\Omega^{g}$ at the discrete level\\
     \hline
     $P$ & Total density matrix on the global supercell
     $\Omega^{g}$ at the discrete level\\
     \hline
     $\CS$ & Collection of indices for the selected columns on the unit
     cell $\Omega^{u}$ \\
     \hline
     $\CS^{g}$ & Collection of all indices for the selected columns on the
     global supercell $\Omega^{g}$ \\
     \hline
     \hline
  \end{tabular}
  \caption{Notation used in the paper}
  \label{tab:nomen}
\end{table}

\begin{figure}[ht!]
\centering
\includegraphics[width=.5\textwidth]{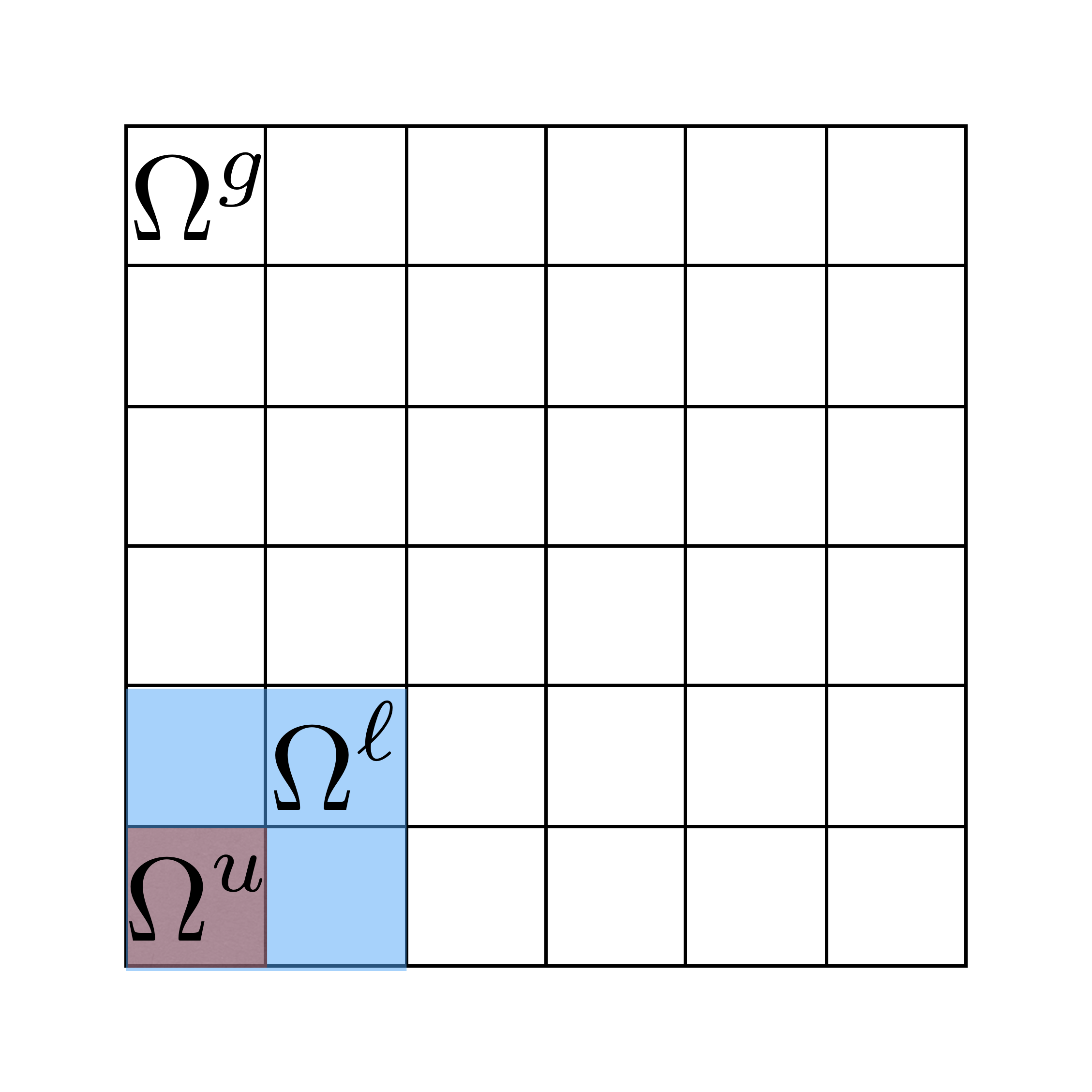}
\caption{\label{fig:kpoint} Illustration of the relationship in two
  dimension ($d=2$) between the unit cell (shaded red), local
  supercell (shaded blue), and global supercell corresponding to $N_1
  = N_2 = 6$ and $N_1^{\ell} = N_2^{\ell} = 2.$}
\end{figure}

\subsection{\Rone{Sub-selection of matrices}}\label{sec:linalg}
Because we deal with very large matrices that exhibit structure due to
the underlying problem set up, it is very useful for us to associate
quantities from the physical problem with portions of matrices.
Therefore, we use set subscripts to denote sub-selection of rows and
columns of matrices. For example, $A_{\{1,2\},\{3,4\}}$ is a submatrix
of $A$ consisting of the intersection of rows one and two with columns
three and four. We use ``$:$'' as a subscript to denote that all rows or columns are considered, \ie~$A_{:,1}$ denotes the first column of $A$.

\subsection{\Rone{Column-pivoted QR factorizations}} \label{sec:QRCP}
Because they play a central role in the development of our algorithms,
we briefly introduce column-pivoted factorizations. Consider $A \in
\mathbb{R}^{n_b\times M}$, where the sizes have been chosen to
coincide with the matrices we will actually be performing these
factorizations on later. A QR factorization with column pivoting (QRCP)
of $A$ follows the algorithm in \cite{businger1965linear} to compute an
$M \times M$ permutation matrix $\Pi,$ a $n_b \times n_b$ orthogonal
matrix $Q$, a $n_b \times n_b$ upper triangular matrix $R$ and a $n_b
\times (M - n_b)$ matrix $T$ such that
\[
A\Pi = Q \begin{bmatrix} R & T\end{bmatrix}.
\]

The column pivoting algorithm in \cite{businger1965linear} is a greedy procedure to try and ensure that $R$ is as well conditioned as possible and that its singular values do not differ too much from those of $A$. If we let $\CS$ denote the original indices of the $n_b$ columns permuted to the front by $\Pi$ then we observe that $A_{:,\CS} = QR$. Hence, if $R$ is well conditioned we expect these columns to form a well conditioned basis for the range of $A$. Finally, What we have presented here is a very narrow definition of such factorizations, and we direct the reader to \cite{chandrasekaran1994rank} and \cite{gu1996efficient} for a more through treatment of such algorithms. 

\section{Kohn-Sham density functional theory with Brillouin zone
sampling} \label{sec:KS}

In this section we provide a relatively self-contained description of
Kohn-Sham DFT, and focus particularly on Brillouin
zone sampling due to the periodic structure. A more detailed discussion
can be found in, $\eg\text{,}$~\cite{Martin2004}.

\subsection{Continuous formulation}

For a crystalline solid modeled by a global supercell $\Omega^{g}$
consisting of $\NS$ unit cells with each unit cell containing $2n_{b}$
electrons (the factor of two comes from spin), the Kohn-Sham equations are~\cite{KohnSham1965} 
\begin{equation}
  H \psi_{\alpha}(\vr) = -\frac12 \nabla^2 \psi_{\alpha}(\vr) + V(\vr) \psi_{\alpha}(\vr) =
  \varepsilon_{\alpha} \psi_{\alpha}(\vr),\quad \vr\in \Omega^{g}, \quad
  \alpha=1,\ldots,n_{b}\NS.
  \label{eqn:KS}
\end{equation}
Each eigenfunction $\psi_{\alpha}$ satisfies the Born-von Karman (BvK)
boundary condition, which is the periodic boundary condition on
$\Omega^{g}$
\begin{equation}
  \psi_{\alpha}(\vr+N_{i}L_{i}\ve_{i}) = \psi_{\alpha}(\vr), \quad \forall \vr\in \Omega^{g},
  \quad i=1,2,3.
  \label{eqn:BvK}
\end{equation}
Using the BvK boundary condition, all eigenvalues
$\{\varepsilon_{\alpha}\}$ are real, and all eigenfunctions
$\{\psi_{\alpha}\}$ are orthogonal to each other under the $L^{2}$ inner
product. We assume the eigenvalues are ordered in a non-descending
manner. 

Kohn-Sham DFT requires solving for the $n_{b}\NS$
eigenfunctions associated with the lowest eigenvalues
$\{\varepsilon_{\alpha}\}_{\alpha=1}^{n_{b}\NS}$.
$\varepsilon_{\alpha}$ is called a Kohn-Sham orbital energy, and
$\psi_{\alpha}$ is called a Kohn-Sham orbital or a Kohn-Sham
wavefunction.  In Kohn-Sham DFT, $V(\vr)$ is the self-consistent
single particle potential, and self-consistency is usually reached
through an iterative procedure. Here without loss of generality we
assume self-consistency has been reached.  We also assume a
pseudopotential is used so $V(\vr)$ is smooth and can be discretized
using uniform grid points. Nonlocal pseudopotentials are neglected for
simplicity of notation, and do not introduce any extra numerical difficulty when
added.  We refer readers to~\cite{Martin2004} for more detailed
explanation of the terminology.

For crystalline solids, $V(\vr)$ is a periodic function in $\Omega^{u}$, \ie~
\begin{equation}
  V(\vr+L_{i}\ve_{i}) = V(\vr), \quad \forall \vr\in \Omega^{g},\quad
  i=1,2,3.
  \label{eqn:Vperiod}
\end{equation}
The Bloch theory (or Bloch-Floquet theory) states that the $n_{b}\NS$ Kohn-Sham
wavefunctions can be relabeled using two indices $\alpha=(b,\vk)$, so
that $\psi_{\alpha}\equiv \psi_{b,\vk}$ can
be decomposed into the form
\begin{equation}
  \psi_{b,\vk}(\vr) = e^{\I \vk\cdot \vr} u_{b,\vk}(\vr),
  \label{eqn:Bloch}
\end{equation}
where $u_{b,\vk}(\vr)$ is the periodic part of $\psi_{b,\vk}(\vr)$ satisfying
\begin{equation}
  u_{b,\vk}(\vr+L_{i}\ve_{i}) = u_{b,\vk}(\vr), \quad \forall \vr\in \Omega^{g},
  \quad i=1,2,3.  
  \label{}
\end{equation}
Using the Bloch decomposition, Eq.~\eqref{eqn:KS} can
be written in terms of $u$ on each unit cell $\Omega^{u}$ as
\begin{equation}
  -\frac12 (\nabla+i\vk)^2 u_{b,\vk}(\vr) + V(\vr) u_{b,\vk}(\vr) =
  \varepsilon_{b,\vk} u_{b,\vk}(\vr), \quad \vr\in \Omega^{u},
  \label{eqn:KSBloch}
\end{equation}
Where each $\vk$ is a
point in the first Brillouin zone defined as 
\[
\mathcal{B} = \left[-\frac{\pi}{L_{1}},\frac{\pi}{L_{1}}\right]\times
\left[-\frac{\pi}{L_{2}},\frac{\pi}{L_{2}}\right]\times
\left[-\frac{\pi}{L_{3}},\frac{\pi}{L_{3}}\right].
\]
For each $\vk$-point, $b=1,\ldots,n_{b}$, \ie~$n_{b}$ is the number of
wavefunctions per $\vk$-point. For simplicity we will drop the $b$ subscript when
describing properties that hold for each $u_{b,\vk}(\vr)$,
$b=1,\ldots,n_{b}$. There are a few $\vk$-points in the Brillouin zone
that play special roles in crystallography and also in numerical
computation. The most important one is the $\Gamma$ point, which is the
origin of the Brillouin zone.
 
In order to solve Eq.~\eqref{eqn:KSBloch}, the Brillouin zone
$\mathcal{B}$ needs to be discretized.  One of the most widely used
discretization schemes is the so-called Monkhorst-Pack
grid~\cite{MonkhorstPack1976}, which uses a uniform discretization of
$\mathcal{B}$. The discretized set of $\vk$-points is
\begin{equation}
  \KS = \left\{ \left( \frac{2\pi j_{1}}{N_{1}L_{1}}, \frac{2\pi
  j_{2}}{N_{2}L_{2}},
  \frac{2\pi j_{3}}{N_{3}L_{3}}\right) + \vs \Big\vert j_{i}=-\frac{N_{i}}{2}+1,\ldots,
  \frac{N_{i}}{2},\quad i=1,2,3 \right\},
  \label{eqn:Kgrid}
\end{equation}
where $\vs$ is a shift vector, and we assume $N_{i}$ is an even
number.  Two common choices are $\vs=(0,0,0)$ (no shift) and
$\vs=\left(
\frac{\pi}{N_{1}L_{1}},\frac{\pi}{N_{2}L_{2}},\frac{\pi}{N_{3}L_{3}}
\right)$ (half grid shift). It should be noted that the inclusion of a
non-zero shift vector could violate the BvK boundary condition. But this
only adds an optional post-processing procedure for handling the phase
vector and will be discussed in section~\ref{subsec:postproc}.  For
now on we assume $\vs=(0,0,0)$, and the BvK boundary condition holds
because
\[
\psi_{b,\vk}(\vr+N_{i}L_{i}\ve_{i})=e^{\I \vk\cdot (\vr+N_{i}L_{i}\ve_{i})}
u_{b,\vk}(\vr+N_{i}L_{i}\ve_{i}) = e^{\I \vk \cdot N_{i}L_{i}\ve_{i}}
\psi_{b,\vk}(\vr) = \psi_{b,\vk}(\vr).
\]
The last equality holds because
\[
e^{\I \vk \cdot N_{i}L_{i}\ve_{i}} = 1,
\]
which is satisfied for $\vk\in \KS$.

\subsection{Discrete formulation}

For each $\vk\in \KS$, Eq.~\eqref{eqn:KSBloch} is solved numerically
for $b = 1,\ldots,n_b$, and we assume the resulting eigenfunctions are
solved for and represented on a uniform grid discretizing the unit
cell $\Omega^{u}$
\begin{equation}
  \RS = \left\{ \left( \frac{j_{1}L_{1}}{M_{1}}, \frac{j_{2}L_{2}}{M_{2}},
  \frac{j_{3}L_{3}}{M_{3}}\right) \Big\vert j_{i}=0,\ldots,
  M_{i}-1,\quad i=1,2,3 \right\}.
  \label{eqn:Rgrid}
\end{equation}
Then each eigenfunction $u_{b,\vk}(\vr)$ \Rtwo{is represented as a column vector}. With some abuse of notation, this column vector is still denoted
by $u_{b,\vk}\in \CC^{\MS\times 1}$, and $u_{b,\vk}(\vj)\equiv
u_{b,\vk}(\vr_{\vj}), \vr_{\vj}\in \RS$.

Since the effectiveness of the technique presented in this paper is heavily
based on numerical linear algebra procedures such as QRCP factorizations and
discrete Fourier transforms, it turns out that using discrete variable
indices such as $\vj$ is more convenient than the continuous variable
indices such as $\vr$.  Therefore we will use discrete indices when
possible for the rest of the paper.  Furthermore, we let
\[
\Omega^{u}
= \left\{(j_1,j_2,j_3)\vert j_{i}=0,\ldots, M_{i}-1 ,\quad i=1,2,3
\right\}
\]
denote the set of indices corresponding to real space grid points
$\RS$ in the unit cell. The periodic boundary
condition on $\Omega^{u}$ allows us to interpret $\vj$ and
$\vj+M_{i}\ve_{i}$ as equivalent points $(i=1,2,3)$.  The periodic
eigenvector $u_{b,\vk}$ satisfies the discrete orthonormal condition
\begin{equation}
  \sum_{\vj\in \Omega^{u}}u^{*}_{b,\vk}(\vj) u_{b',\vk'}(\vj) = \delta_{b,b'}
  \delta_{\vk,\vk'}.
  \label{eqn:Uortho}
\end{equation}
Again, we denote by
\begin{equation}
\Omega^{g}=\left\{(j_1,j_2,j_3)\vert j_{i}=0,\ldots, N_{i}M_{i}-1 ,\quad
i=1,2,3 \right\}
  \label{eqn:Omegag}
\end{equation}
the corresponding set of indices of real space grid points in the
global supercell.  Similar to before, the periodic boundary condition on
$\Omega^{g}$ allows us to interpret $\vj$ and $\vj+N_{i}M_{i}\ve_{i}$
as equivalent points $(i=1,2,3)$.  The discretized eigenfunction
$\psi_{b,\vk}(\vj)$ is periodic on the global supercell $\Omega^{g}$,
and satisfies the discrete orthonormal condition
\begin{equation}
  \sum_{\vj\in \Omega^{g}}\psi^{*}_{b,\vk}(\vj) \psi_{b',\vk'}(\vj) = \delta_{b,b'}
  \delta_{\vk,\vk'} \NS.
  \label{eqn:Uortho_psi}
\end{equation}
Note that the convention taken in Eq.~\eqref{eqn:Omegag} places the
origin of the unit cell $\Omega^{u}$ at the origin of the global
supercell $\Omega^{g}$ as well. This is allowed due to the periodicity
of the global supercell.

\Rtwo{We now introduce a key concept for the development of our algorithm, the discrete density matrix.} Notationally, it is denoted by $P_{\vk}\in
\CC^{(\MS\NS)\times (\MS\NS)}$ and for each $\vk$-point is defined as
\begin{equation}
  P_{\vk}(\vj,\vj') = \frac{1}{\NS} \sum_{b=1}^{n_{b}}
  \psi_{b,\vk}(\vj)\psi^*_{b,\vk}(\vj'), \quad \vj,\vj'\in
  \Omega^{g}.
  \label{eqn:Pk}
\end{equation}
Here $*$ stands for the Hermitian conjugate operation.  
The complete discrete density matrix is then defined as
\begin{equation}
  P(\vj,\vj') = \sum_{\vk\in \KS} P_{\vk}(\vj,\vj'), \quad \vj,\vj'\in
  \Omega^{g}.
  \label{eqn:Ptot}
\end{equation}
It is easy to verify that $P_{\vk}$ and $P$ satisfies the normalization
conditions
\[
\Tr P_{\vk} = n_{b} \quad \text{and} \quad \Tr P = \NS n_{b}.
\]

The following block circulant property of the density matrix plays an
important role in the SCDM-k method for constructing localized
orbitals.
\begin{prop}
  $P$ satisfies the block circulant property, i.e.
  \[
  P(\vj+M_{i}\ve_{i},\vj'+M_{i}\ve_{i}) = P(\vj,\vj'),\quad \forall
  \vj,\vj'\in \Omega^{g}, \quad i=1,2,3.
  \]
  \label{prop:Ptrans}
\end{prop}
\begin{proof}
  It is sufficient to show that each $P_{\vk}$ satisfies the block
  circulant property.
  For any $i=1,2,3$,
  \[
  \begin{split}
  &P_{\vk}(\vj+M_{i}\ve_{i},\vj'+M_{i}\ve_{i}) 
  = \frac{1}{\NS} \sum_{b=1}^{n_{b}}
  \psi_{b,\vk}(\vj+M_{i}\ve_{i})\psi^*_{b,\vk}(\vj'+M_{i}\ve_{i})\\
  = & \frac{1}{\NS} \sum_{b=1}^{n_{b}}
  e^{\I \vk\cdot (\vr_{\vj}+L_{i}\ve_{i})}u_{b,\vk}(\vj+M_{i}\ve_{i}) 
  e^{-\I \vk\cdot
  (\vr_{\vj'}+L_{i}\ve_{i})}u^{*}_{b,\vk}(\vj'+M_{i}\ve_{i}) \\
  = & 
  e^{ \I \vk\cdot \vr_{\vj}} u_{b,\vk}(\vj) 
  e^{-\I \vk\cdot \vr_{\vj'}}u^{*}_{b,\vk}(\vj')  = P_{\vk}(\vj,\vj').
  \end{split}
  \]
  Therefore, each $P_{\vk}$ is block circulant. Here we have
  implicitly used the aforementioned periodic structure over
  $\Omega^g$ to address when $\vj+M_{i}\ve_{i}$ or $\vj'+M_{i}\ve_{i}$
  yields a point outside the boundary of $\Omega^g$.
\end{proof}

\section{Selected columns of the density matrix}
\label{sec:scdm}

\Rtwo{Once the Kohn-Sham equations have been solved numerically, we have the means to construct a set of $\NS n_b$ eigenfunctions encoded as columns of $\Psi^g$ over the global supercell $\Omega^g$. However, the functions will be delocalized spatially. We now outline a construction for computing $\NS n_b$ localized eigenfunctions over the global supercell that span the same space as $\Psi^g$. Notably, the matrix of all $\NS n_b$ vectors over $\NS M$ spatial points may be prohibitively expensive to even store. As a consequence of this, we provide algorithms that construct $n_b$ of these localized functions associated with a single unit cell. The periodic structure of the problem means this is sufficient for our needs.}

\subsection{\Rtwo{Review of the SCDM procedure for $\Gamma$ point calculation}}
In order to present the SCDM-k method, we first briefly review the
procedure for localizing Kohn-Sham orbitals via the SCDM procedure for
$\Gamma$ point calculations, of which the details 
can be found in Ref.~\cite{SCDM}.  In order to remain notationally
consistent within this work, we use slightly different notation than
in \cite{SCDM}. 

We present the SCDM method as if Kohn-Sham orbitals are
only defined on a single unit cell $\Omega^{u}$ with the one $\vk$
point, the so-called the $\Gamma$ point, \ie~$\NS = 1$. Let
$\left\{\psi_\alpha\right\}_{\alpha=1}^{n_b}$ represent the $n_b$
Kohn-Sham orbitals discretized on a uniform grid, and collected as
columns of the matrix $\Psi$. We leverage the fact that the density matrix $P =\Psi\Psi^*$ has well localized columns for insulating systems \cite{BenziBoitoRazouk2013}, and use columns of $P$ as a starting point for constructing a localized basis. We are always interested in finding a representation of the entire Kohn-Sham invariant subspace and thus construct $n_b$ localized orbitals.

Algorithm~\ref{alg:scdmgamma_simple} presents the SCDM algorithm in its simplest form, providing a unitary transform from $\Psi$ to a set of orthogonal localized orbitals $\Phi$. Here we see that the algorithm computationally amounts to a single QRCP factorization. This factorization can be computed, $\eg\text{,}$ via the \texttt{qr} function in MATLAB \cite{MATLAB} or the DGEQP3 routine in LAPACK \cite{lapack}.

To motivate the algorithmic developments here, we also present a slight
variation on Algorithm~\ref{alg:scdmgamma_simple}. Specifically, we
assume that we do not have access to the matrix Q. Instead we simply
have a set of $n_b$ columns that the permutation matrix $\Pi$ chose to
move forward during the QRCP process. It turns out that this information
is sufficient to generate a localized basis. This variation is presented
in Algorithm~\ref{alg:scdmgamma}, where we first compute the set $\CS$
and then construct the relevant columns of $\Psi\Psi^*.$ These columns
are themselves well localized but they are not orthogonal, which is a
desirable property. Fortunately, because $$\Psi\Psi =
P_{:,\CS}\left(P_{\CS,\CS}\right)^{-1}P_{:,\CS}^*$$ we may orthogonalize
$P_{:,\CS}$ using the square root of $\left(P_{\CS,\CS}\right)^{-1}$.
The rapid decay away from the diagonal of $P_{\CS,\CS}$ ensures that the
resulting orthogonalized columns remain well localized.

\begin{algorithm}  
\begin{small}
\begin{center}
  \begin{minipage}{5in}
\begin{tabular}{p{0.5in}p{4.5in}}
{\bf Input}:  &  \begin{minipage}[t]{4.0in}
  The orthonormal Kohn-Sham orbitals $\Psi.$
\end{minipage} \\
{\bf Output}:  &  \begin{minipage}[t]{4.0in}
  The orthogonalized SCDM $\Phi.$
\end{minipage} 
\end{tabular}
\begin{algorithmic}[1]
  \STATE Compute the QRCP factorization $\Psi^*\Pi = QR$
  \RETURN $\Phi = \Psi Q$
\end{algorithmic}
\end{minipage}
\end{center}
\end{small}
\caption{\Rtwo{Simple algorithm for SCDM with $\Gamma$-point calculation.}}
\label{alg:scdmgamma_simple}
\end{algorithm}

\begin{algorithm}  
\begin{small}
\begin{center}
  \begin{minipage}{5in}
\begin{tabular}{p{0.5in}p{4.5in}}
{\bf Input}:  &  \begin{minipage}[t]{4.0in}
  The Kohn-Sham orbitals $\Psi.$
\end{minipage} \\
{\bf Output}:  &  \begin{minipage}[t]{4.0in}
  The SCDM $\tilde{\Phi}$ or the orthogonalized SCDM $\Phi.$
\end{minipage} 
\end{tabular}
\begin{algorithmic}[1]
	\STATE Compute the QRCP factorization $\Psi^*\Pi = QR$
  \STATE Let $\CS$ denote the original indices of the first $n_b$ columns selected by  $\Pi$.
  \STATE Compute the SCDM $\tilde{\Phi} = P_{:,\CS} = \Psi
  (\Psi_{\CS,:})^{*} $, which are localized orbitals.
  \IF{the orthogonalized SCDM are desired}
  \STATE Compute $\left(P_{\CS,\CS}\right)^{-1/2}$.
  \STATE Compute the orthonormal orbitals as $\Phi = \tilde{\Phi}\left(P_{\CS,\CS}\right)^{-1/2}$.
  \RETURN $\Phi$
  \ELSE
  \RETURN $\tilde{\Phi}$
  \ENDIF
\end{algorithmic}
\end{minipage}
\end{center}
\end{small}
\caption{\Rtwo{Algorithm for SCDM with $\Gamma$-point calculation.}}
\label{alg:scdmgamma}
\end{algorithm}

\subsection{SCDM with Brillouin zone sampling}

\Rtwo{When we had a single $\vk$-point we sought to compute a set of
$n_b$ localized functions in the cell associated with that $\vk$-point.
Now, we may view our spatial domain as consisting of $N$ unit cells,
consequently we seek to compute $\NS n_b$ localized functions over the
entire spatial domain $\Omega^g$. The most straightforward way to
accomplish this would be to simply treat a larger problem, one with $\NS
M$ spatial unknowns and $\NS n_b$ eigenfunctions denoted by $\Psi^{g}$,
as input to Algorithm~\ref{alg:scdmgamma_simple} or~\ref{alg:scdmgamma}.
In this case, the global density matrix
$\Psi^{g}\left(\Psi^{g}\right)^*$ would exhibit the locality we desire.
However, this could be prohibitively expensive since the QRCP computation would scale as $\Or(\NS^3)$. However, conceptually such a procedure can serve as a point of comparison for the performance of our new algorithm.}

\Rtwo{To overcome this obstacle, we appeal to the block circulant
property in Proposition~\ref{prop:Ptrans}. Applying our algorithm to
$\Psi^{g}$ directly we would expect that the index set $\CS$, of size
$\NS n_b$, would contain $n_b$ columns associated with each unit cell
$\Omega^{u}$. Therefore, we will solve a smaller problem to extract
$n_b$ columns corresponding to  a single unit cell. We then construct a set denoted $\CS^g$ that may be used for the entire global supercell by simply adding the translates of these columns into the other unit cells, with $\NS$ unique translates of the $n_b$ columns this yields a set of the desired size. Once we have this set of columns of the global density matrix, we leverage the block circulant structure of $P$ to efficiently perform the orthogonalization in a manner analogous to lines 6 and 7 of Algorithm~\ref{alg:scdmunit}.}

The set $\CS^g$ constructed by the above procedure does not necessarily coincide with the columns of $P$ we would select if we used our existing algorithm directly on the global problem. However, we make the assumption that while the detailed shape
of the columns of the density matrix requires a relatively large
number of $\vk$-points to resolve, the actually selection of columns in a given unit cell is relatively insensitive to the number of $\vk$-points used. Since
$u_{b,\vk}$ is discretized on a fine real space grid, even if the columns selected 
shift by a few grid points the resulting columns of the global
density matrix should still be well localized.  Numerical experiments in
section~\ref{sec:numer} corroborate this intuitive argument.

We now explicitly introduce the small local supercell $\Omega^{\ell}$ where
$\Omega^{u}\subset \Omega^{\ell}\subset \Omega^{g}$
(Figure~\ref{fig:kpoint}) used in to pick the $n_b$ columns of $\CS$ in a single unit cell. With a similar abuse of notation as before, $\Omega^{\ell}$ is also used to denote the indices of
grid points in the local supercell. The number of unit cells in
$\Omega^{\ell}$ along the $i$-th direction is denoted by $N_{i}^{\ell}$.
Following the same convention as in the definition of the global
supercell in Eq.~\eqref{eqn:Omegag}, the grid points in $\Omega^{\ell}$ are
\begin{equation}
  \Omega^{\ell}=\left\{(j_1,j_2,j_3)\vert j_{i}=0,\ldots,
  N^{\ell}_{i}M_{i}-1 ,\quad i=1,2,3 \right\},
  \label{eqn:Omegal}
\end{equation}
which places the
origin of the unit cell $\Omega^{\ell}$ also at the origin of the global
supercell $\Omega^{g}$.

\subsection{Computing columns of the density matrix}
We now discuss the construction of selected columns of the density matrix that are localized. Let $q_{i}=N_{i}/N_{i}^{\ell}$. If $q_{i}$ is an integer for
$i=1,2,3$ then the Kohn-Sham equations on the local supercell
$\Omega^{\ell}$ can be 
simply solved through a coarse sampling of the Brillouin zone. 
The resulting collection of grid points in the Brillouin zone is denoted by
\begin{equation}
  \KS^{\ell} = \left\{ \left( \frac{2\pi j_{1}}{N_{1}^{\ell}L_{1}}, \frac{2\pi
  j_{2}}{N_{2}^{\ell}L_{2}},
  \frac{2\pi j_{3}}{N_{3}^{\ell}L_{3}}\right) \Big\vert
  j_{i}=-\frac{N_{i}^{\ell}}{2}+1,\ldots,
  \frac{N_{i}^{\ell}}{2},\quad i=1,2,3 \right\}.
  \label{eqn:Kgridlocal}
\end{equation}
Since the local supercell is only used as a numerical tool \Rtwo{to select columns} more efficiently, we make an additional
approximation by discarding the shift vector $\vs$ when reconstructing
the Kohn-Sham orbital $\psi^{\ell}_{b,\vk}$ from its periodic part
$u_{b,\vk}$, and all $\psi^{\ell}_{b,\vk}$'s satisfying the BvK boundary condition in
$\Omega^{\ell}$. 

\Rtwo{After solving the Kohn-Sham equations on the local supercell via coarse sampling of the Brillouin zone, we get $\NS^\ell n_b$ orthonormal wavefunctions and arrange them into the columns of the $\NS^\ell M \times \NS^\ell n_b$ matrix $\Psi^{\ell}$.}  We now apply the SCDM procedure for
$\Gamma$-point calculations as outlined in Algorithm~\ref{alg:scdmgamma} to
$\Psi^{\ell}$, which selects $n_b \NS^{\ell}$ \Rtwo{columns denoted $\CS^\ell$}. We then restrict
this larger set of selected columns to the set $\CS^u$ which contains the
$n_{b}$ columns associated with points inside a single unit cell
$\Omega^{u}$. 

Given $\CS^u$, we may compute the respective selected columns of the density matrix on the entire
global supercell as
$P(\vj,\vc),\vj\in \Omega^{g},\vc\in \CS^u$. We first construct
$P_{\vk}(\vj,\vc),\vj\in \Omega^{u},\vc\in \CS^u$ as
\begin{equation}
   P_{\vk}(\vj,\vc) = \frac{1}{\NS} \sum_{b=1}^{n_{b}}
   e^{\I \vk\cdot(\vr_{j}-\vr_{c})}u_{b,\vk}(\vj)u^*_{b,\vk}(\vc).
   \label{eqn:PkOmegau}
\end{equation}
Then for $\vj\in \Omega^{g}\backslash \Omega^{u}$, note that for any
$n_{i}=0,\ldots,N_{i}-1, i=1,2,3$, we have
\begin{equation}
  \begin{split}
    P_{\vk}(\vj+n_{i}M_{i}\ve_{i},\vc) = &
    e^{\I \vk\cdot (n_{i} L_{i}\ve_{i})} \left( \frac{1}{\NS}
    \sum_{b=1}^{n_{b}} e^{\I \vk\cdot (\vr_{j}-\vr_{c})}
    u_{b,\vk}(\vj + n_{i}M_{i}\ve_{i}) 
    u_{b,\vk}^*(\vc)
    \right)\\
    = & e^{\I \vk\cdot (n_{i} L_{i}\ve_{i})} P_{\vk}(\vj,\vc).
  \end{split}
  \label{eqn:PkOmegag}
\end{equation}
Therefore, $P_{\vk}$ can be constructed just by multiplying
$P_{\vk}(\vj,\vc)$, $\vj\in \Omega^{u}$ by phase factors.  Summing up
$P_{\vk}(\vj,\vc)$'s for all $\vk$ we obtain $P(\vj,\vc)$. For
convenience we let $P_{\CS}\in \CC^{(\MS\NS)\times n_{b}}$ denote the
matrix elements
\[
P_{:,\CS^u} \left(\vj,b\right) = P(\vj,\vc_{b}), \quad \vj\in
\Omega^{g}, \vc_{b}\in \CS^u.
\]
The preceding discussion is summarized in Algorithm \ref{alg:scdmunit}
and yields the desired selected columns of the density matrix over the
global supercell. \Rtwo{Note that this corresponds to computing only $n_b$ of the $\NS n_b$ total localized functions we expect. If desired, the others may be constructed by translating $\CS^u$ into a different unit cell, see the following section for details.}

\begin{algorithm}  
\begin{small}
\begin{center}
  \begin{minipage}{5in}
\begin{tabular}{p{0.5in}p{4.5in}}
{\bf Input}:  &  \begin{minipage}[t]{4.0in}
  Monkhorst-Pack points in the Brillouin zone $\KS$;\\
  Periodic parts of Kohn-Sham orbitals $\{u_{n\vk}\}$ for $n =
  1,\ldots,n_b$ and $\vk\in \KS$;\\
  Sub-sampling Monkhorst-Pack points in the reciprocal space 
  $\KS^{\ell}$;\\
\end{minipage} \\
{\bf Output}:  &  \begin{minipage}[t]{4.0in}
  Non-orthogonal SCDM associated with the unit cell $\Omega^u$.
\end{minipage} 
\end{tabular}
\begin{algorithmic}[1]
	\STATE Construct $\Psi^\ell$ from $\{u_{n\vk}\}$ with $\vk \in
  \KS^{\ell}$ from the Bloch decomposition~\eqref{eqn:Bloch}.
  \STATE Compute the QRCP factorization $\left(\Psi^\ell\right)^*\Pi = QR$
  \STATE Let $\CS^\ell$ denote the original indices of the first $\NS^\ell n_b$ columns selected by  $\Pi$.
  \STATE Find the selected column indices $\CS^u = \left\{j \in \CS^\ell
  \vert j \in \Omega^u\right\}$ in the unit cell.
  \FORALL{$\vk$}
  \STATE Construct $P_{\vk}(\vj,\vc)$ for $\vj \in \Omega^u$ and $\vc \in \CS^u$ via \eqref{eqn:PkOmegau}.
  \ENDFOR
  \STATE Compute $P(\vj,\vc)$ for $\vj \in \Omega^g$ and $\vc\in\CS^u$ using \eqref{eqn:PkOmegag}.
\end{algorithmic}
\end{minipage}
\end{center}
\end{small}
\caption{\Rtwo{Computing the (non-orthogonal) selected columns of the density
matrix inside a unit cell.}}
\label{alg:scdmunit}
\end{algorithm}

\subsection{Construct the orthonormalized SCDM}

\Rtwo{We now have a procedure to compute $\CS^u$ and} 
due to the block circulant property of the density matrix, this is
sufficient.  All the remaining columns are the
block translates of these columns into the other unit cells in
$\Omega^{g}$.  Define the collection of indices
\[
\CS^{g} \equiv \left\{ \vc + (n_{1}M_{1},n_{2}M_{2},n_{3}M_{3}) \vert
\vc\in \CS,\quad n_{i}=0,\ldots,N_{i}-1,\quad i=1,2,3 \right\},
\]
which induces a matrix $P_{:,\CS^{g}}\in \CC^{(\MS\NS)\times
  (n_{b}\NS)}$ such that
\[
P_{:,\CS^{g}}\left(\vj,b\right) = P(\vj,\vc_{b}), \quad \vj\in
\Omega^{g}, \vc_{b}\in \CS^{g}.
\]
\Rtwo{These columns of $P$ are precisely the $\NS n_b$ columns that form
our localized basis over the entire global supercell. However, they are
not orthonormal and we now describe an efficient procedure to orthonormalize them.}

The matrix $P_{:,\CS^{g}}$ is block circulant when viewed as an $\NS\times \NS$ block
matrix with each block of size $\MS\times n_{b}$.
Note that the storage cost of $P_{\CS^{g}}$ is $\Or(n_{b}\MS\NS^2)$,
and it is therefore never explicitly computed or stored. We also
define the matrix $P_{\CS^{g},\CS^{g}}\in \CC^{(n_{b}\NS)\times
  (n_{b}\NS)}$ as
\[
P_{\CS^{g},\CS^{g}} \left( a,b \right) = P \left( \vc_{a},\vc_{b} \right), 
\quad \vc_{a}, \vc_{b}\in \CS^{g}
\]
and note that $P_{\CS^{g},\CS^{g}}$ is also block circulant when viewed as an
$\NS\times \NS$ block matrix with each block of size $n_{b}\times
n_{b}$.

Now, if we consider the matrix $\Phi^g\in \CC^{(\MS\NS)\times
(n_{b}\NS)}$ defined as
\begin{equation}
  \Phi^g = P_{:,\CS^{g}} \left( P_{\CS^{g},\CS^{g}} \right)^{-\frac12},
  \label{eqn:orthopsig}
\end{equation}
it is also block circulant and satisfies the discrete orthonormality condition
\[
\left(\Phi^g\right)^{*}\Phi^g = I.
\]
Therefore, $\Phi^g$ represents an orthonormal set of localized basis functions
across all of the $\vk$-points. Due to the translational invariance, we
only need to compute the columns of $\Phi^g$ centered in $\Omega^{u}$,
and as before, the remaining columns are just the block translates of these columns
into the other unit cells in $\Omega^{g}$. 
Similar to $P_{:,\CS^{g}}$, the entire matrix $\Phi^g$ is
neither explicitly constructed nor stored in practice.

In fact, all of the rows of $P_{:,\CS^g}$ and $\Phi^g$ may be
constructed from knowledge of the first $\MS$ rows, \ie~those
associated with a single unit cell. More specifically, if
$\vj \in \Omega^u$ and $\vj' \in \Omega^g$ are such that
\[
\vj' = \vj + (n_1M_1,n_2M_2,n_3M_3)
\]
for some $(n_1,n_2,n_3)$, then for any $\vc_{b}\in \CS^g$
\begin{equation}
\label{eqn:circ}
\Phi^g(\vj',\vc_{b}) = \Phi^g\left(\vj,\vc_{b}-(n_1M_1,n_2M_2,n_3M_3)\right).
\end{equation}
To make this
explicit notationally, let $P_{\Omega^u,\CS^g}\in \CC^{(\MS)\times
  (n_{b}\NS)}$ be such that
\[
P_{\Omega^u,\CS^{g}} \left( \vj,b \right) = P(\vj,\vc_{b}), \quad \vj\in
\Omega^{u}, \vc_{b}\in \CS^{g}.
\]
We define $\Phi^u$ as
\[
\Phi^u = P_{\Omega^u,\CS^{g}} \left( P_{\CS^{g},\CS^{g}} \right)^{-\frac12},
\]
and since we may construct $\Phi^g$ from $\Phi^u$ via
Eq.~\eqref{eqn:circ} we now focus on the construction of
$\Phi^u$.

Direct computation of the matrix square root of
$P_{\CS^{g},\CS^{g}}\in \CC^{(n_{b}\NS)\times(n_{b}\NS)}$ in
Eq.~\eqref{eqn:orthopsig} costs $\Or(\NS^{3})$ and is hence
computationally expensive.  Instead, we demonstrate an algorithm to
take advantage of the block circulant property of
$P_{\CS^{g},\CS^{g}}$ that reduces
the complexity to $\Or(\NS \log \NS)$.  Let $\FS$ be the matrix representing the 
three-dimensional discrete Fourier transform that acts with respect to the $\vn
= (n_1,n_2,n_3)$ index in the set $\CS^{g}$ and
$\FS^{-1} = \FS^*$ be the three-dimensional inverse discrete Fourier transform that
acts with respect to the Fourier space grid corresponding to
$\vn$. Computationally these operations are performed via the fast Fourier transform. Observe that,
\begin{equation}
  \begin{split}
    \Phi^u = &
    P_{\Omega^u,\CS^{g}} \FS^{-1} \FS
    \left(P_{\CS^{g},\CS^{g}}\right)^{-\frac12} \FS^{-1} \FS \\
    = & \left(P_{\Omega^u,\CS^{g}} \FS^{-1}\right) \left( \FS P_{\CS^{g},\CS^{g}}
    \FS^{-1} \right)^{-\frac12} \FS.
  \end{split}
  \label{eqn:psigcompute}
\end{equation}
\Rtwo{We can move the square root outside $F$ and $F^{-1}$ because $\FS \left(P_{\CS^{g},\CS^{g}}\right)^{-\frac12} \FS^{-1}$ and
$\left( \FS P_{\CS^{g},\CS^{g}} \FS^{-1} \right)^{-\frac12}$ have exactly the same eigenvalues and eigenvectors. This is a consequence of the fact that $F$ is unitary and $P_{\CS^{g},\CS^{g}}$ is Hermitian positive definite.}

Eq.~\eqref{eqn:psigcompute} compactly represents the procedure
for accelerating the computation of the orthogonalized SCDM and we now
elaborate on their construction. First, we observe that the matrix
\begin{equation}
\hat{P}_{\CS^{g},\CS^{g}} = \FS P_{\CS^{g},\CS^{g}}
\FS^{-1}
\label{eqn:hatp_diag}
\end{equation}
is block diagonal with $\NS$ blocks each of size $n_{b}\times n_{b}$.
This means that, $\left(\hat{P}_{\CS^{g},\CS^{g}}\right)^{-1/2}$ may
be computed by taking the inverse square root of each small block
independently. The block diagonal structure follows immediately from
the fact that $P_{\CS^{g},\CS^{g}}$ is block circulant and the use of the discrete Fourier transform. Furthermore,
the diagonal blocks may be computed by applying $\FS$ to the first
$n_b$ columns of $P_{\CS^{g},\CS^{g}}$,
which means those are the only columns we need to explicitly construct. Finally, since there are $\NS$ diagonal blocks in
total, we may associate each of the diagonal blocks with an index
$\vn$.  Therefore, we refer to a given diagonal block of
$\hat{P}_{\CS^{g},\CS^{g}}$ as $\hat{P}_{\CS^{g},\CS^{g};\vn}$.

We now let
\begin{equation}
\hat{P}_{\Omega^u,\CS^g} = P_{\Omega^u,\CS^{g}} \FS^{-1},
\label{eqn:hatp}
\end{equation} 
which may be constructed by applying $\FS$ to the columns of
$\left(P_{\Omega^u,\CS^{g}}\right)^*$. Finally, similar to before we let $\hat{P}_{\Omega^u,\CS^g;\vn}$ denote a $\MS \times n_b$
matrix containing $n_b$ of the columns of $\hat{P}_{\Omega^u,\CS^g}$ associated
with a given index $\vn$. We define  
\begin{equation}
\hat{\Phi}^u = \Phi^u \FS^{-1},
\label{eqn:hatphi}
\end{equation} 
and once again let $\hat{\Phi}^u_{\vn}$ denote a $\MS \times n_b$
matrix containing $n_b$ columns of $\hat{\Phi}^u$ associated with a
given index $\vn$.

The use of a single index allows us to compactly write the
construction of $\hat{\Phi}^u$ as
\begin{equation}
\hat{\Phi}^u_{\vn} = \hat{P}_{\Omega^u,\CS^g;\vn}\left(\hat{P}_{\CS^{g},\CS^{g};\vn}\right)^{-1/2}.
\end{equation}
We may then form the first $\MS$ rows associated with a single unit cell
of $\Phi^g$ by applying $\FS^{-1}$ to the columns of
$\hat{\Phi}^u_{\vn}$. Algorithm~\ref{alg:orthoscdm} summarizes the steps for constructing the
orthogonalized SCDM centered in a single unit cell.

\begin{algorithm}  
\begin{small}
\begin{center}
  \begin{minipage}{5in}
\begin{tabular}{p{0.5in}p{4.5in}}
{\bf Input}:  &  \begin{minipage}[t]{4.0in}
  Monkhorst-Pack points in the Brillouin zone $\{\vk\}$;\\
  Non-orthogonal selected columns of the density matrix
  $P_{\CS}$;\\
\end{minipage} \\
{\bf Output}:  &  \begin{minipage}[t]{4.0in}
  Orthogonal SCDM $\Phi^{g}$ associated with the unit cell
  $\Omega^u$.
\end{minipage} 
\end{tabular}
\begin{algorithmic}[1]
  \STATE Compute the first $n_b$ columns of $P_{\CS^{g},\CS^{g}}$, and 
  $P_{\Omega^u,\CS^{g}}$.
  \STATE Compute $\hat{P}_{\CS^{g},\CS^{g}}$ via
  Eq.~\eqref{eqn:hatp_diag}.
  \STATE Compute $\left(
  \hat{P}_{\CS^{g},\CS^{g};\vn}\right)^{-\frac12}$ for all $\vn$.
  \STATE Compute $\hat{P}_{\Omega^u,\CS^g} = \left(\FS
  \left(P_{\Omega^u,\CS^g}\right)^*\right)^*$.
  \STATE Compute $\hat{\Phi}^u_{\vn} =
  \hat{P}_{\CS^u;\vn}\left(\hat{P}_{\CS^{g},\CS^{g};\vn}\right)^{-1/2}$
  for all $\vn$.
  \STATE Compute $\Phi^u = \left(\FS^{-1}  \left(\Phi^u\right)^*\right)^*$.
  \STATE Construct $\Phi^g(:,\CS)$ from $\Phi^u$ via
  Eq.~\eqref{eqn:circ}.
\end{algorithmic}
\end{minipage}
\end{center}
\end{small}
\caption{Computing the orthonormalized SCDM from the non-orthogonal SCDM.}
\label{alg:orthoscdm}
\end{algorithm}

\subsection{Post-processing for the shift vector in the Monkhorst-Pack
grid}\label{subsec:postproc}

Here we address the issues that arise when using the Monkhorst-Pack
grid with half grid shift $\vs$.  In such case, each orbital
$\psi_{\alpha}$
does not satisfy the BvK boundary condition in the global supercell
$\Omega^{g}$.  Instead,
\begin{equation}
  \psi_{\alpha}(\vr+N_{i}L_{i}\ve_{i}) = -\psi_{\alpha}(\vr), \quad \forall \vr\in \Omega^{g},
  \quad i=1,2,3,
  \label{eqn:BvK2}
\end{equation}
and a phase factor $(-1)$ is gained.  To be accurate, this phase
factor needs to be taken into account in the SCDM.  Note that
\begin{equation}
  \widetilde{P}(\vj,\vc) = \sum_{\vk} e^{\I \vs \cdot
  (\vr_{\vj}-\vr_{\vc})} P_{\vk}(\vj,\vc) = e^{\I \vs \cdot
  (\vr_{\vj}-\vr_{\vc})} P(\vj,\vc).
  \label{eqn:postproc}
\end{equation}
Here $P_{\vk}$ and $P$ are the density matrices obtained by taking
$\vs=(0,0,0)$.  Therefore the post-processing only requires multiplying
each column of the SCDM $P(\vj,\vc)$ and the orthonormalized SCDM
$\Phi^g_{\vj,\vc}$ by a phase vector $e^{\I \vs \cdot
(\vr_{\vj}-\vr_{\vc})}$.  It is straightforward to verify that the
post-processing procedure~\eqref{eqn:postproc} maintains the
orthonormality of the orthonormalized SCDM.

\subsection{Complexity}

The computational complexity for selecting the SCDM using the QRCP factorization with a local supercell is $\Or((\NS^{\ell})^3\MS
n_{b}^2)$.  Then, the complexity for computing the non-orthonormal SCDM
$P_{\Omega^g,\CS^u}$ via matrix-matrix multiplication is $\Or(\MS\NS n_{b}^2)$.  Due to
the usage of Eq.~\eqref{eqn:PkOmegag} the cost for computing
${P_{\CS}(\vj,\vc),\vj\in \Omega^{g}}$ is only $\Or(\MS\NS n_{b})$. \Rtwo{Importantly,
$\MS$ and $n_{b}$ are assumed to be fixed with respect to
the increase of the number of $\vk$-points (i.e. the number of unit
cells contained in the global supercell). Furthermore, we explicitly set $\NS^\ell$ to be small and not grow with $\NS$. Therefore, the cost for obtaining the
non-orthogonal SCDM is $\Or(\NS)$.}

In order to generate the orthonormal SCDM, the cost for computing the
Fourier transform $\FS P_{\CS^{g},\CS^{g}} \FS^{-1}$ is $\Or(\NS \log
(\NS) n_{b}^2)$, and the cost for computing the matrix square root of
the block diagonal matrix $\FS P_{\CS^{g},\CS^{g}} \FS^{-1}$ is
$\Or(\NS n_{b}^3)$.  The cost for computing the first block row of
$\left(P_{\CS^{g}} \FS^{-1}\right)$ is $\Or(\NS \log (\NS) \MS n_{b})$
and the cost for multiplying with matrix square root is $\Or(\NS \MS
n_{b}^2)$.  Therefore, the total cost for generating
$\Phi^g(\vj,\vc),\vj\in \Omega^{g},\vc\in \Omega^{u}$ is
\[
\Or(\NS \log
(\NS) n_{b}^2+\NS n_{b}^3+\NS \log (\NS) \MS n_{b}+\NS \MS n_{b}^2)).
\]
Finally, the cost for the post-processing by multiplying a phase vector
is $\Or(\MS\NS n_{b})$.  If we take the leading term with respect to
$\NS,\MS$ and think of $n_{b}$ as a small constant, then the
complexity of the whole algorithm is
\[
\Or\Big(\NS \log (\NS) \MS n_{b} + \NS \MS n_{b}^2 + (\NS^{\ell})^3\MS
n_{b}^2\Big).
\]
Hence, the complexity for obtaining the orthonormalized SCDM with
respect to the number of $\vk$-points used for sampling the Brillouin zone is
only $\Or(\NS\log \NS)$.

\section{Numerical examples}\label{sec:numer}

To illustrate the performance of the SCDM-k algorithm, we consider the
localization of orbitals obtained from model potentials in two and three
dimensions. In three dimensions, the model
potential takes the form
\begin{equation}
  V(\vr) = \sum_{n_{1}=0}^{N_{1}-1} \sum_{n_{2}=0}^{N_{2}-1}\sum_{n_{3}=0}^{N_{3}-1}
  V_{u}\left(\vr-\sum_{i=1}^{3}n_{i}\ve_{i}\right).
  \label{eqn:modelpotential}
\end{equation}
Here $V_{u}(\vr)$ is taken to be a Gaussian potential centered at the
origin in the unit cell $\Omega^{u}$ modeling an atom, \ie~

\begin{equation}
V_{u}(\vr) = -4.0 e^{-\frac{\norm{\vr}^2}{2\sigma^2}}.
\end{equation}
\Rtwo{The shape of the model potential in two dimensions is similar,
\begin{equation}
  V(\vr) = \sum_{n_{1}=0}^{N_{1}-1} \sum_{n_{2}=0}^{N_{2}-1}
  V_{u}\left(\vr-\sum_{i=1}^{2}n_{i}\ve_{i}\right).
  \label{eqn:modelpotential2}
\end{equation}}
We generally set $\sigma = 1$ and will explicitly state when we use a
different value.

\subsection{Shapes of the SCDM}
We first present numerical examples illustrating the shapes of the
SCDM in two and three dimensions.  In the two dimensional case, we set
$L_{1}=L_{2}=6.0$, $M_1=M_2=32$, $n_{b}=3$, and use \Rtwo{$N_1=N_2=8$ $\vk$-points per
direction.} For this example we set $\sigma = 0.8.$ Setting $n_{b}=3$
means that we should have one orbital that behaves similar to an
$s$-orbital \Rtwo{(spherical)} and two orbitals that behave similar to a
$p$-orbital \Rtwo{(non-spherical)}. Figure \ref{fig:plot2d} shows the shape of SCDM and the
orthonormalized SCDM. We only plot the three SCDM for a single
$\vk$-point near the middle of the domain and observe the rapid decay
of both the SCDM and the orthogonalized SCDM away from their unit
cell.

In three dimensions, we use \Rtwo{$N_1=N_2=N_3=4$} $\vk$-points per direction, set
$M_1=M_2=M_3=20$, and let $n_b=4$. Similar to the 2D case, we expect
that there should be one orbital that behaves similar to an
$s$-orbital and three orbitals that behave similar to a
$p$-orbital. Figure \ref{fig:plot2d} shows the SCDM and the
orthogonalized SCDM. \Rtwo{Here we only plot isosurfaces for the four SCDM for a single
$\vk$-point near the middle of the domain. As in the two dimensional case, we see that the bulk of the orbital is well localized within a single unit cell. It is also even more apparent than in the 2d case that our localized functions match the expected $s$- and $p$-orbital structure.}

In both of the preceding cases, the well localized orbitals also imply
that the matrix $P_{\CS^{g},\CS^{g}}$ and hence the matrix $\FS
P_{\CS^{g},\CS^{g}}\FS^{-1}$ are well conditioned, and the computation
of the matrix square root does not cause any numerical problems. In
fact, the condition number of the small block matrices we have to take
the inverse square root of was less than five in the two dimensional
example and 15 for the three dimensional example.

\subsection{Locality of the SCDM}
Figure~\ref{fig:scdm2d} and~\ref{fig:scdm3d} show that the SCDM are
qualitatively very well localized.  We now quantify this by measuring
the locality of the functions systematically. Specifically, we
construct the SCDM and the orthogonalized SCDM and then, over the
global supercell $\Omega^g$, measure the fraction of entries (denoted by $nz\%$)
where the relative magnitude of the functions is above a given threshold
$\epsilon.$

\begin{figure}[ht!]

  \centering \subfloat[]{\includegraphics[width=\textwidth]{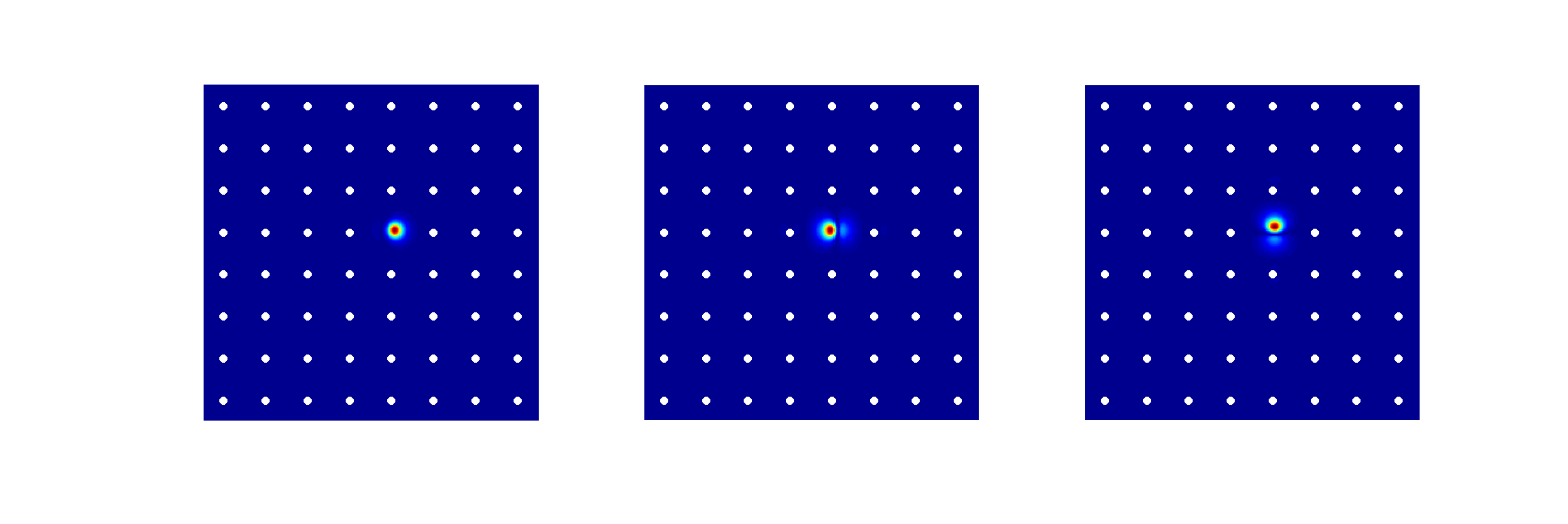}}

\subfloat[]{\includegraphics[width=\textwidth]{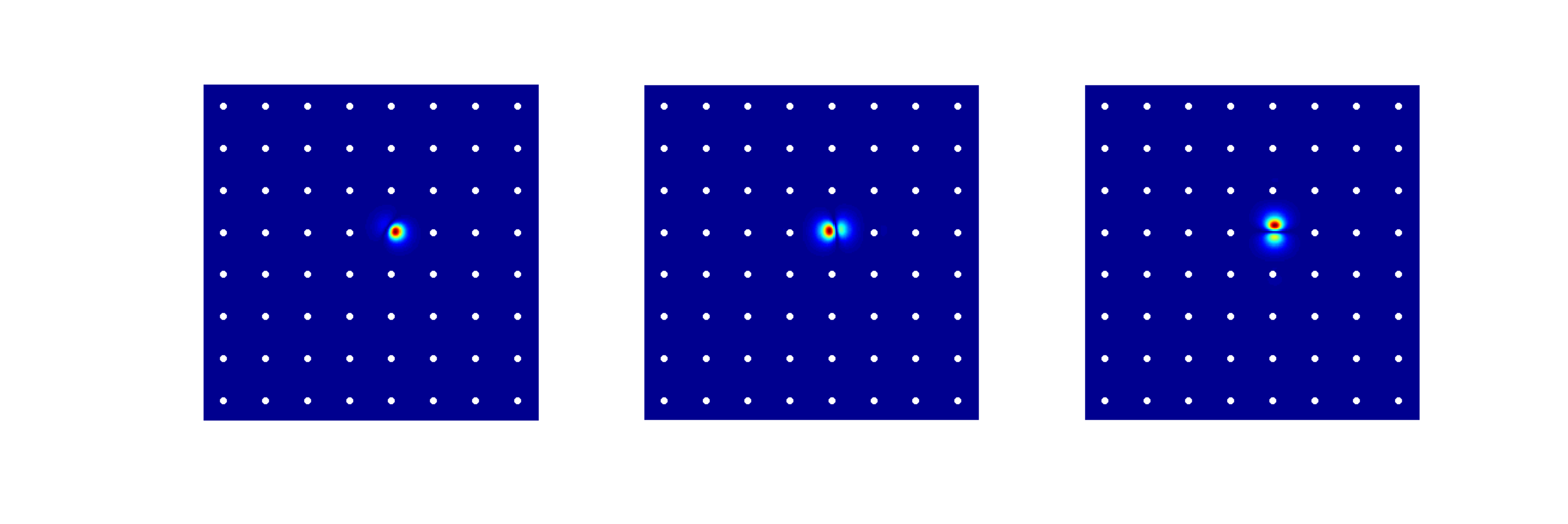}}

\subfloat[]{\includegraphics[width=\textwidth]{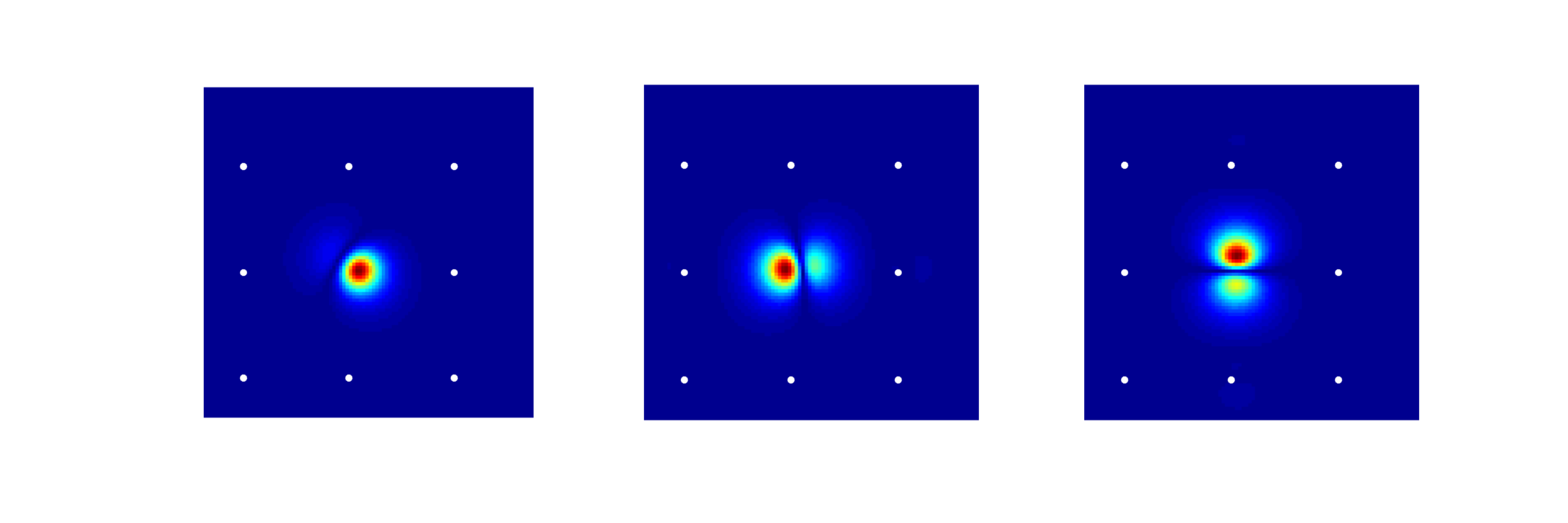}}

\caption{\label{fig:plot2d} Absolute value of the three SCDM (a) and
orthogonalized SCDM (b) located in a single unit cell and plotted over
the global supercell. \Rtwo{Zoomed-in images of the significant regions
of orthogonalized SCDM (c) showing the expected spherical ($s$-orbital
like) and non-spherical ($p$-orbitals like) structure.}}
\label{fig:scdm2d}
\end{figure}

\begin{figure}[ht!]
\centering
\subfloat[]{\includegraphics[width=\textwidth]{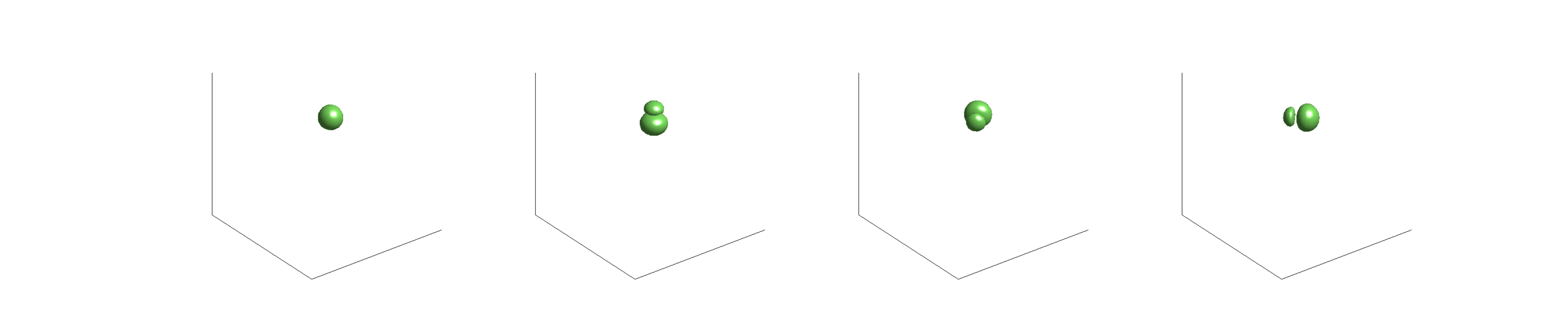}}

\subfloat[]{\includegraphics[width=\textwidth]{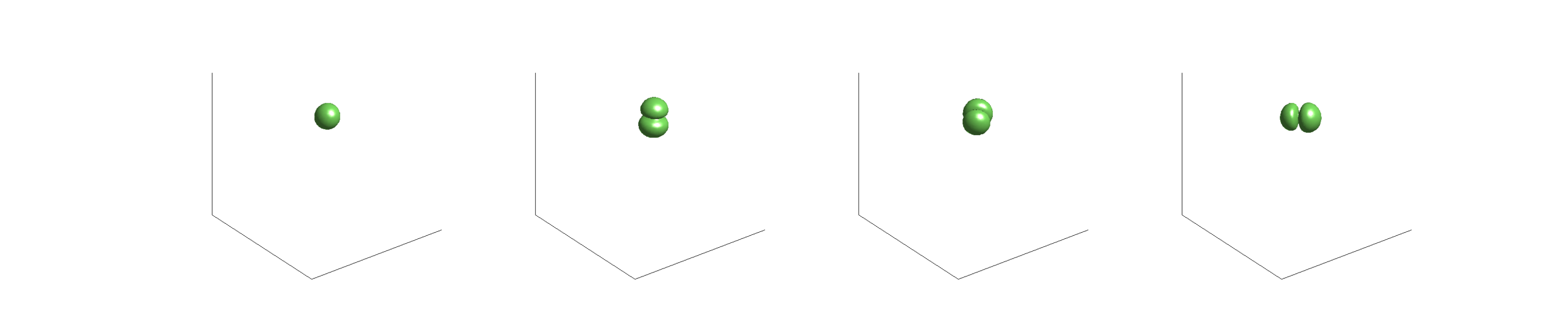}}
\caption{\label{fig:plot3d} Isosurface at a relative value of $0.1$ of
  the absolute value of the SCDM (a) and orthogonalized SCDM (b)
  located in a single unit cell and plotted over the global
  supercell.}
\label{fig:scdm3d}
\end{figure}

First, we measure the locality for the two dimensional problem. We use
\Rtwo{$N_1 = N_2 = 16$} $\vk$-points in each direction, $n_b=3$, and $M_1 = M_2 = 40.$
Figure \ref{fig:2dnz} shows the average locality for both the SCDM and
the orthogonalized SCDM. The diameter of the localized region is
proportional to the square root of the volume, and thus is
proportional to $(nz\%)^{1/2}$, consequently this is the quantity
we choose to plot. Importantly, here we see that the orthogonalization does not severely
impact the localization properties of the SCDM. Furthermore, when the
relative truncation threshold is set to $\epsilon = 10^{-2}$ for the
orthogonalized SCDM less than $1\%$ of the entries remain
non-zero. While in some cases here the orthogonalized orbitals are
actually more localized, we do not necessarily expect such behavior in
general and only expect that the orthogonalization step will not
significantly reduce the locality.

\begin{figure}[ht!]
\centering
\includegraphics[width=.75\textwidth]{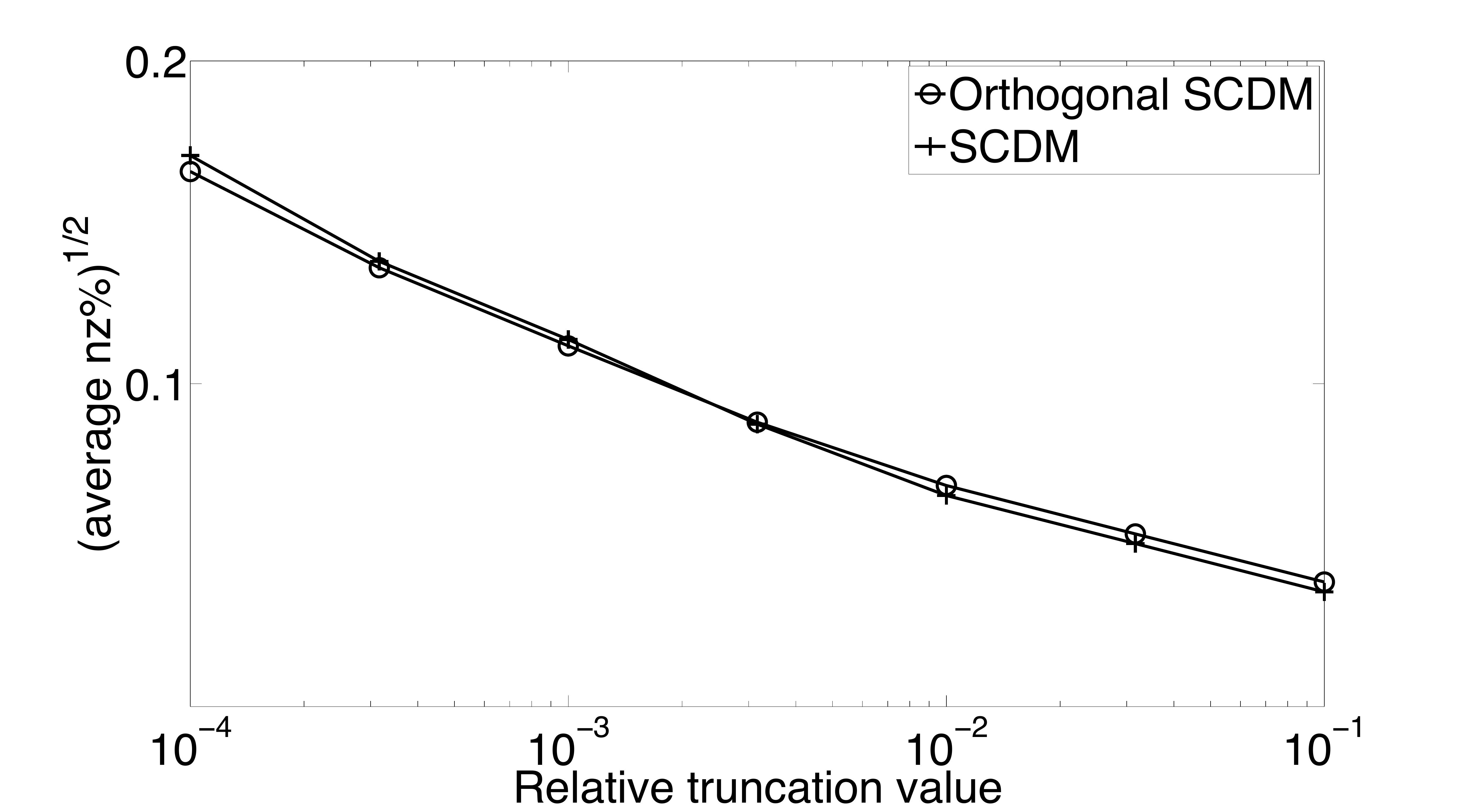}
\caption{\label{fig:2dnz} Average fraction of nonzero entries after
  truncation for the SCDM and the orthogonalized SCDM in two
  dimensions.}
\end{figure}

\Rtwo{We now move back to the three dimensional problem and use
$N_1 = N_2 = N_3 = 8$ $\vk$-points in each direction, $n_b = 4$, and $M_1 = M_2 = M_3 = 20.$
In Figure \ref{fig:3dnz} we plot average locality for both the SCDM and
the orthogonalized SCDM. Analogously to before, the diameter of the localized region is
proportional to the cube root of the volume, and thus we choose to plot $(nz\%)^{1/3}$.
Once again, the orthogonalization does not severely impact the
localization properties of the SCDM and when the relative
truncation threshold is set to $10^{-2}$ only
about $0.7\%$ of the entries or the orthogonalized SCDM remain non-zero.}

\begin{figure}[ht!]
\centering
\includegraphics[width=.75\textwidth]{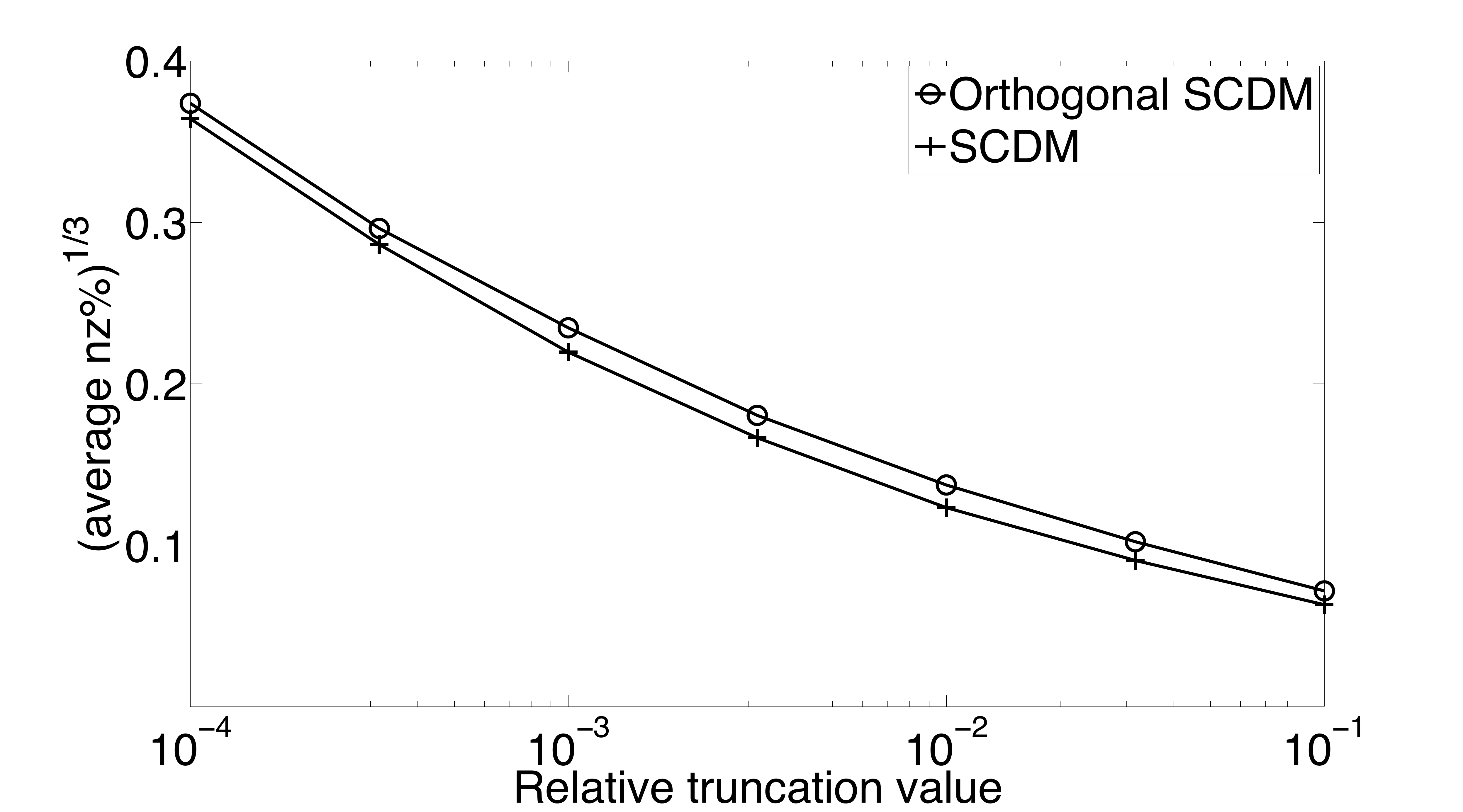}
\caption{\label{fig:3dnz} Average fraction of nonzero entries after
  truncation for the SCDM and the orthogonalized SCDM in three
  dimensions.}
\end{figure}

\subsection{Changing the local super cell size}
Rather than finding the selected columns using
wavefunctions defined on the entire global supercell $\Omega^g$, we find a good
approximation to these columns by using a much smaller local
supercell $\Omega^\ell$, whose size does not increase as the size of the global
supercell grows. \Rtwo{In fact, this is the key approximation made by our algorithm, and the only source of ``error'' when compared to our existing methods.}   Here we quantitatively study the dependence of the
quality of the columns selected by this local supercell approach. 

Concretely, we use a two dimensional problem with ${N_1 = N_2 = 16}$
$\vk$-points in each direction, $n_b = 3$, and $M_1 = M_2 = 20.$ We
proceeded to vary the size of the local supercell, \Rtwo{$N_1^\ell$ and $N_2^\ell,$} and compute the SCDM
and orthogonalized SCDM. We measure the locality as the fraction of
non-zero entries after truncation at a relative magnitude of
$10^{-2}.$ Figure \ref{fig:local} shows that the localization of both
the SCDM and the orthogonalized SCDM is nearly constant as the local
supercell size varies.

\begin{figure}[ht!]
\centering
\includegraphics[width=.75\textwidth]{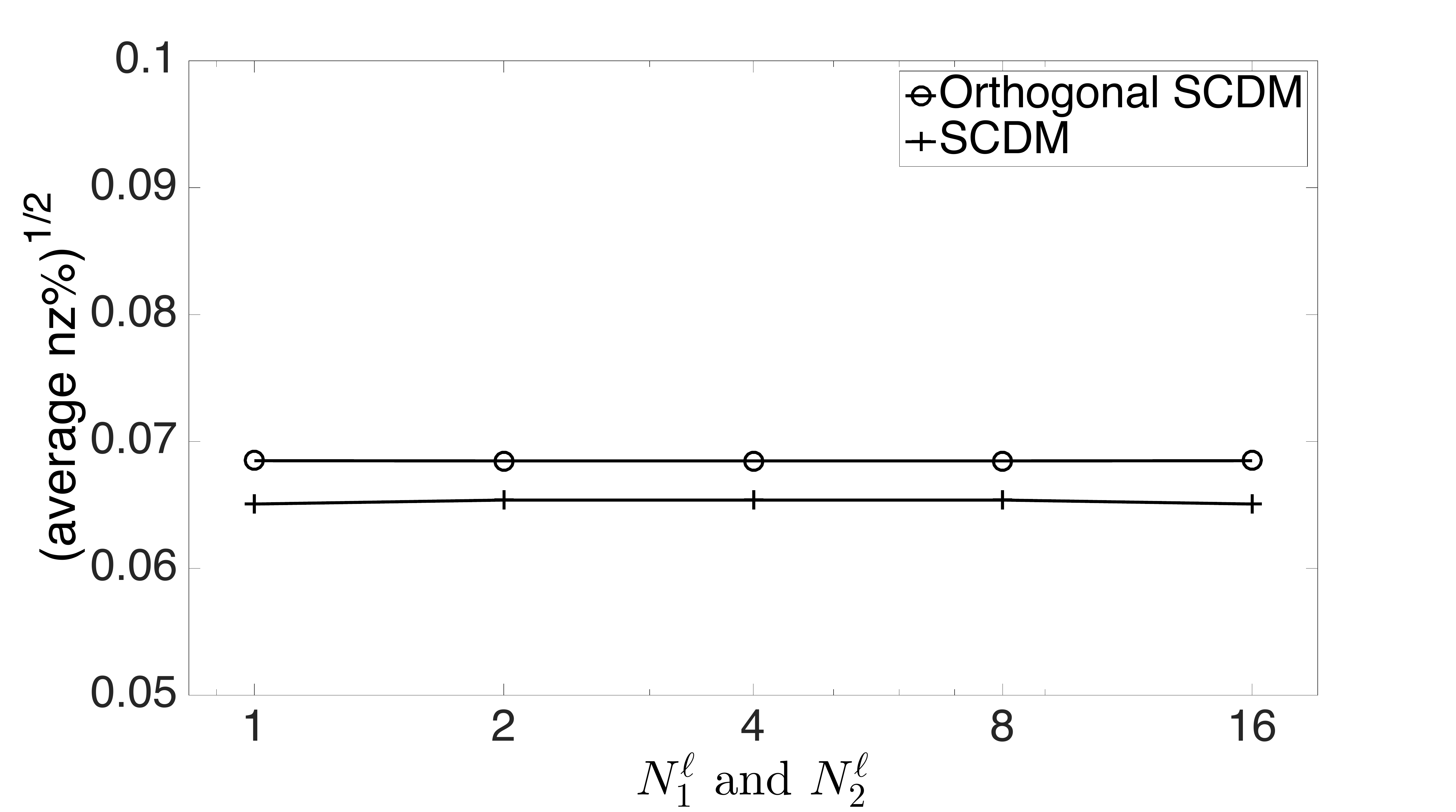}
\caption{\label{fig:local} locality as local supercell size grows}
\end{figure} 

\Rtwo{To compare against our existing methods and validate our approximation, we let the local supercell size grow to match that of the global supercell. This corresponds to running Algorithm~\ref{alg:scdmgamma} on the global problem and in Figure~\ref{fig:local} occurs at the rightmost point of the plot since $N_1=N_2 = 16$ and $N_1^\ell = N_2^\ell = 16$. Therefore, we observe that there is no noticable error introduced by our new algorithm: we get functions that are just as localized as if we treated the global problem directly.}

\subsection{Scaling with the number of $\vk$-points}

Finally, we demonstrate the computational scaling outlined in section
\ref{sec:scdm}. Here we consider a three dimensional problem and 
increase the total number of $\vk$-points in each direction. In this
experiment we used $M_1 = M_2 = M_3 = 10$ and $n_b = 4.$ Figure \ref{fig:3dscale} shows the time taken to
compute the orthogonalized SCDM as the number of $\vk$-points grows. 
In
this case, the terms linear in $\NS$ actually dominate the computation
and we observe close to linear scaling.
\begin{figure}[ht!]
\centering
\includegraphics[width=.75\textwidth]{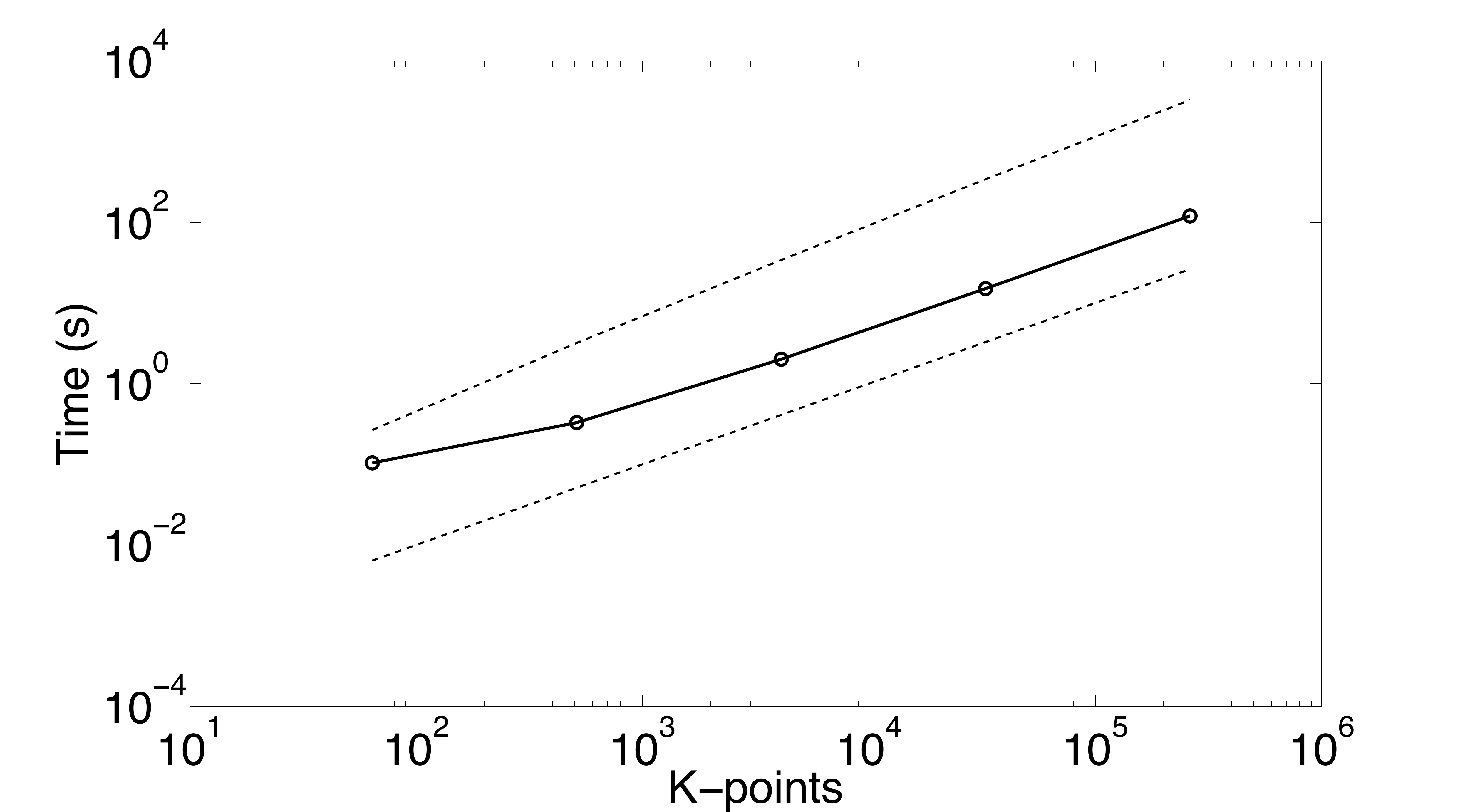}
\caption{\label{fig:3dscale} Time taken to compute the orthogonalized
SCDM as the number of $\vk$-points grows. The upper dotted line
represents $\Or(\NS \log \NS)$ scaling and the lower dotted line
represents $\Or(\NS)$ scaling}
\end{figure}

\section{Conclusion and future work}\label{sec:conclusion}

We developed the SCDM-k method, which is a new method for finding both
orthogonal and non-orthogonal localized orbitals from a set of Kohn-Sham
orbitals obtained from Brillouin zone sampling. The SCDM-k method is
implicitly based on the use of the gauge invariant density matrix, and
obtains localized orbitals without an iterative optimization procedure.
Furthermore, the computation exhibits $\Or(\NS \log \NS)$ scaling with
respect to the total number of $\vk$-points used. Numerical results for
two and three dimensional systems with model potentials indicate that
the SCDM-k method generates localized orbitals that can be visually 
similar to MLWFs.  All routines used in the SCDM-k method are standard
linear algebra routines and are thus easily parallelizable. This could
enable efficient computation of localized orbitals for solids both in
the post-processing step and performed on the fly. Though described using
a uniform real space grid for example, the SCDM-k method does not rely
on a particular basis set and can be combined with any electronic
structure software packages supporting $\vk$-point sampling in the
Brillouin zone.  We plan to apply the SCDM-k method to compute localized
orbitals and compare directly with MLWFs for Kohn-Sham DFT calculations
of real materials systems in the near future.

\section*{Acknowledgments}

The work of A.D. is partially supported by a Simons Graduate Research
Assistantship. The work of L.L. is partially supported by the DOE
Scientific Discovery through the Advanced Computing (SciDAC) program,
the DOE Center for Applied Mathematics for Energy Research
Applications (CAMERA) program, and by an Alfred P. Sloan
fellowship. The work of L.Y. is partially supported by the National
Science Foundation under grant DMS-0846501 and the DOE's Advanced
Scientific Computing Research program under grant
DE-FC02-13ER26134/DESC0009409. The authors thank Eric Bylaska, Sinisa
Coh, Felipe da Jornada and Bert de Jong for useful discussions, and the anonymous referees for their help in improving this manuscript.

\bibliographystyle{elsarticle-num} 
\bibliography{localizeband}

\begin{thebibliography}{10}
\expandafter\ifx\csname url\endcsname\relax
  \def\url#1{\texttt{#1}}\fi
\expandafter\ifx\csname urlprefix\endcsname\relax\def\urlprefix{URL }\fi
\expandafter\ifx\csname href\endcsname\relax
  \def\href#1#2{#2} \def\path#1{#1}\fi

\bibitem{HohenbergKohn1964}
P.~Hohenberg, W.~Kohn, {Inhomogeneous electron gas}, Phys. Rev. 136 (1964)
  B864--B871.

\bibitem{KohnSham1965}
W.~Kohn, L.~Sham, {Self-consistent equations including exchange and correlation
  effects}, Phys. Rev. 140 (1965) A1133--A1138.

\bibitem{FosterBoys1960}
J.~M. Foster, S.~F. Boys, Canonical configurational interaction procedure, Rev.
  Mod. Phys. 32 (1960) 300--302.

\bibitem{MarzariVanderbilt1997}
N.~Marzari, D.~Vanderbilt, Maximally localized generalized {W}annier functions
  for composite energy bands, Phys. Rev. B 56~(20) (1997) 12847--12865.

\bibitem{WannierReview}
N.~Marzari, A.~A. Mostofi, J.~R. Yates, I.~Souza, D.~Vanderbilt, Maximally
  localized {W}annier functions: Theory and applications, Rev. Mod. Phys. 84
  (2012) 1419--1475.

\bibitem{WuSelloniCar2009}
X.~Wu, A.~Selloni, R.~Car, Order-{N} implementation of exact exchange in
  extended insulating systems, Phys. Rev. B 79~(8) (2009) 085102.

\bibitem{GygiDuchemin2012}
F.~Gygi, I.~Duchemin, Efficient computation of {H}artree--{F}ock exchange using
  recursive subspace bisection, J. Chem. Theory Comput. 9~(1) (2012) 582--587.

\bibitem{King-SmithVanderbilt1993}
R.~D. King-Smith, D.~Vanderbilt, Theory of polarization of crystalline solids,
  Phys. Rev. B 47 (1993) 1651--1654.

\bibitem{Goedecker1999}
S.~Goedecker, {Linear scaling electronic structure methods}, Rev. Mod. Phys. 71
  (1999) 1085--1123.

\bibitem{UmariStenuitBaroni2009}
P.~Umari, G.~Stenuit, S.~Baroni, Optimal representation of the polarization
  propagator for large-scale {GW} calculations, Phys. Rev. B 79~(20) (2009)
  201104.

\bibitem{UmariStenuitBaroni2010}
P.~Umari, G.~Stenuit, S.~Baroni, {GW} quasiparticle spectra from occupied
  states only, Phys. Rev. B 81 (2010) 115104.

\bibitem{Gygi2009}
F.~Gygi, Compact representations of {K}ohn--{S}ham invariant subspaces, Phys.
  Rev. Lett. 102 (2009) 166406.

\bibitem{ELiLu2010}
W.~E, T.~Li, J.~Lu, Localized bases of eigensubspaces and operator compression,
  Proc. Natl. Acad. Sci. 107~(4) (2010) 1273--1278.

\bibitem{OzolinsLaiCaflischEtAl2013}
V.~Ozoli{\c{n}}{\v{s}}, R.~Lai, R.~Caflisch, S.~Osher, Compressed modes for
  variational problems in mathematics and physics, Proc. Natl. Acad. Sci.
  110~(46) (2013) 18368--18373.

\bibitem{AquilantePedersenMerasEtAl2006}
F.~Aquilante, T.~B. Pedersen, A.~S. de~Mer{\'a}s, H.~Koch, Fast noniterative
  orbital localization for large molecules, J. Chem. Phys. 125~(17) (2006)
  174101.

\bibitem{SCDM}
A.~Damle, L.~Lin, L.~Ying, Compressed representation of {K}ohn--{S}ham orbitals
  via selected columns of the density matrix, J. Chem. Theory Comput. 11~(4)
  (2015) 1463--1469.

\bibitem{Kohn1996}
W.~Kohn, Density functional and density matrix method scaling linearly with the
  number of atoms, Phys. Rev. Lett. 76 (1996) 3168--3171.

\bibitem{ProdanKohn2005}
E.~Prodan, W.~Kohn, {Nearsightedness of electronic matter}, Proc. Natl. Acad.
  Sci. 102 (2005) 11635--11638.

\bibitem{BenziBoitoRazouk2013}
M.~Benzi, P.~Boito, N.~Razouk, Decay properties of spectral projectors with
  applications to electronic structure, SIAM Rev. 55~(1) (2013) 3--64.

\bibitem{Blount}
E.~Blount, Formalisms of band theory, Vol.~13 of Solid State Phys., Academic
  Press, 1962, pp. 305--373.

\bibitem{Cloizeaux1964a}
J.~D. Cloizeaux, Energy bands and projection operators in a crystal: Analytic
  and asymptotic properties, Phys. Rev. 135 (1964) A685--A697.

\bibitem{Cloizeaux1964b}
J.~D. Cloizeaux, Analytical properties of $n$-dimensional energy bands and
  {W}annier functions, Phys. Rev. 135 (1964) A698--A707.

\bibitem{Nenciu}
G.~Nenciu, Existence of the exponentially localised {W}annier functions, Comm.
  Math. Phys. 91~(1) (1983) 81--85.

\bibitem{MonkhorstPack1976}
H.~J. Monkhorst, J.~D. Pack, Special points for {B}rillouin-zone integrations,
  Phys. Rev. B 13 (1976) 5188--5192.

\bibitem{Martin2004}
R.~Martin, Electronic Structure -- Basic Theory and Practical Methods,
  Cambridge Univ. Pr., West Nyack, {NY}, 2004.

\end{thebibliography}

\end{document}